\documentclass[sigconf, nonacm]{acmart}
\usepackage[ruled,vlined,linesnumbered]{algorithm2e}
\usepackage{subcaption}
\usepackage{amssymb}
\usepackage{multirow} 
\usepackage{multicol}
\usepackage{adjustbox}
\usepackage{enumitem}
\usepackage{pifont} 

\begin{document}
\title{Low-Latency Out-of-Core ANN Search in High-Dimensional Space}

\author{Ziwen Song}
\affiliation{%
  \institution{Northeastern University}
  \streetaddress{195 Chuangxinlu}
  \city{Shenyang}
  \state{Liaoning}
  \postcode{110169}
}
\email{songzw1@mails.neu.edu.cn}

\author{Bin Wang}
\affiliation{%
  \institution{Northeastern University}
  \streetaddress{195 Chuangxinlu}
  \city{Shenyang}
  \state{Liaoning}
  \postcode{110169}
}
\email{binwang@mail.neu.edu.cn}

\author{Xiaochun Yang}
\affiliation{%
  \institution{Northeastern University}
  \streetaddress{195 Chuangxinlu}
  \city{Shenyang}
  \state{Liaoning}
  \postcode{110169}
}
\email{yangxc@mail.neu.edu.cn}

\author{Junhua Zhang}
\affiliation{%
  \institution{Northeastern University}
  \streetaddress{195 Chuangxinlu}
  \city{Shenyang}
  \state{Liaoning}
  \postcode{110169}
}
\email{zhangjunhua@neu.edu.cn}

\begin{abstract}
In-memory graph-based approximate nearest neighbor (ANN) search has superior search performance but incurs significant memory footprint. Disk-based methods reduce memory usage but suffer from high disk access latency. A common challenge is how to achieve low-latency search while significantly reducing memory footprint. In this paper, we propose SkipDisk, a disk-memory hybrid ANN search that significantly reduces memory footprint while achieving search latency comparable to or lower than in-memory method HNSW. By analyzing existing disk-based methods, we observed that disk access remains the primary bottleneck, and existing lower bound based filtering methods are two loose to effectively reduce disk access. Therefore, we design SkipDisk to achieve tight lower bound with low memory footprint to reduce the search latency. First, we design a dedicated pivot for each point to improve the lower bound of the triangle inequality for effective filtering. We further design an estimation-based approach based on this lower bound. Second, to reduce the memory footprint, we employ a three-level data pruning strategy to preserve informative data in memory. Third, to further reduce search latency, we design an asynchronous I/O strategy based on the decoupling of in-memory search and disk access by storing neighbor nodes in memory. Experiments show that our method achieves a latency of 85\% of HNSW's latency with approximately 10\% memory footprint, and a latency to 63\% of HNSW's with a slightly higher memory footprint of around 20\%.
\end{abstract}

\maketitle

\section{Introduction}
\begin{figure}[t!]
\begin{subfigure}[t]{.22\textwidth}
  \centering
  \includegraphics[width=\textwidth]{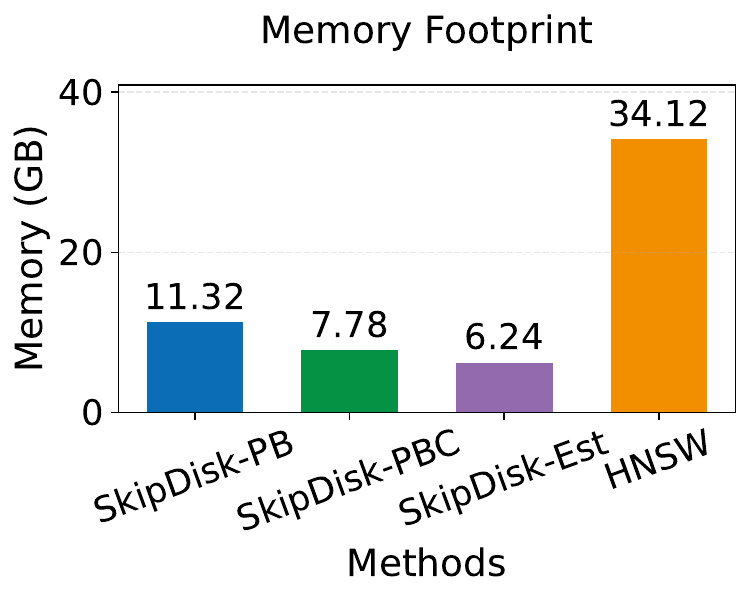}  
  \caption{Memory footprint}
  \label{fig:intrmemory}
\end{subfigure}
\begin{subfigure}[t]{.22\textwidth}
  \centering
  \includegraphics[width=\textwidth]{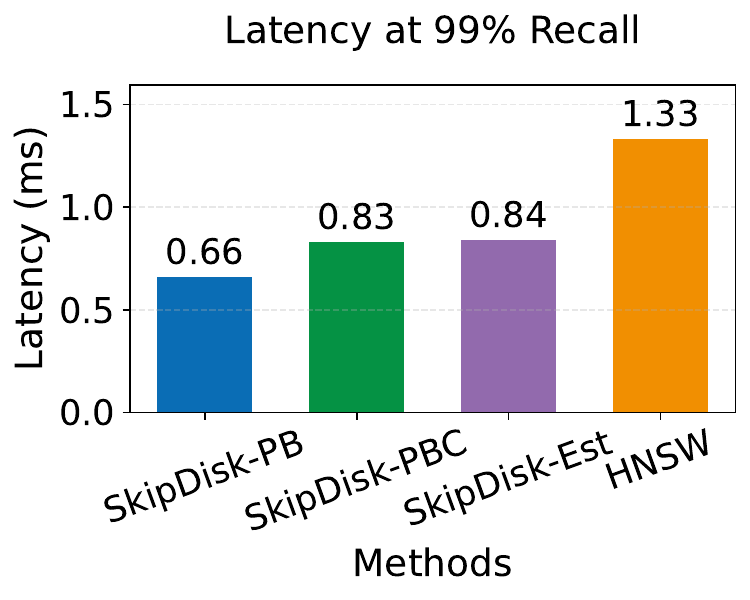} 
  
  \caption{Search latency}
  \label{fig:intrlatency}

\end{subfigure}

\caption{Illustration of memory footprint and search Latency on BIGCODE dataset of our method compared to HNSW. Our three variants of SkipDisk are able to achieve lower latency than HNSW with reduced memory footprint. SkipDisk-PB can achieve a latency of 50\% of HNSW's latency with around 30\% memory footprint.}

\label{fig:intrmemlat}
\Description{Two figure. Left is for memroy footprint. Right is latency.}

\end{figure}

Approximate nearest neighbor (ANN) search on high dimensional space~\cite{indyk1998approximate} is a fundamental problem and plays an important role in the era of artificial intelligence. It powers a wide range of applications, from recommendation systems~\cite{DBLP:journals/www/WuZQWGSQZZLXC24,DBLP:conf/kdd/LiLJLYZWM21} and image retrieval~\cite{DBLP:conf/sigmod/WangYGJXLWGLXYY21} to how large language models retrieve knowledge~\cite{DBLP:conf/kdd/FanDNWLYCL24}. Consequently, developing high-performance and memory efficient ANN search methods has become a critical research area and a fundamental research challenge.

In-menory graph-based ANN search have long dominated the field due to their superior search performance. Among those, the Hierarchical Navigable Small World (HNSW) method~\cite{DBLP:journals/pvldb/LuKXI21} has been widely adopted in vector databases. A key limitation of in-memory methods is the requirement to load the entire data into memory, which poses a significant challenge for large-scale datasets. With the rising price of memory, reducing memory footprint to save costs has become increasingly important. Meanwhile, maintaining high search performance is equally critical. This presents a dilemma: how to maintain low latency search while significantly reducing memory footprint? Even halving memory footprint would yield significant cost savings. This leads us to a central problem we address in this work: \textbf{Is it possible to achieve search latency comparable to, or even surpassing, that of in-memory method HNSW while significantly reducing memory footprint?}

Researchers have proposed disk-based methods, such as DiskANN to reduce the memory footprint of graph-based ANN search~\cite{jayaram2019diskann}. DiskANN stores neighbor information and full-precision data on disk and keeps quantized data in memory. During search, it retrieves neighbor information and original vectors from disk to update and re-rank candidates. Starling~\cite{DBLP:journals/pacmmod/WangXYWPKGXGX24} optimizes the data layout by colocating frequently accessed points in the same disk block and uses asynchronous I/O access on it to improve the performance. PipeANN~\cite{DBLP:journals/corr/abs-2509-25487} use a poll-based asynchronous I/O to fully exploit disk performance and overlap in-memory search with disk. Trim~\cite{DBLP:journals/pacmmod/SongZGYWWQ25} reduce disk access number by using a lower bound based on triangle inequality. Gorgeous~\cite{yin2025gorgeous} reduces disk access by only considering top-ranked candidate points for exact distance computations based on distance ordering in the quantized data. AlayaLaser~\cite{DBLP:journals/corr/abs-2602-23342} uses a SIMD-friendly data layout and a large volume of in-memory data to improve performance.

We conduct an analysis and observe that the high cost of disk access remains the main bottleneck in Disk-based methods. With vector dimensions in embeddings often being 768 or higher, a single 4K disk block can store at most one vector, making the data layout optimization strategy of Starling ineffective. With the large number of disk accesses, PipeANN struggles to fully hide disk access latency. The lower bound computed by Trim is relatively loose, resulting in limited filtering effectiveness. Its p-Relaxed variant filters more points by relaxing the bound, but increasing the degree of relaxation leads to a higher error rate, which degrades recall and makes it difficult to achieve good latency. Gorgeous relies on the precision of the quantized methods and simply limiting the selection to the top-ranked portion may lead to performance degradation due to precision loss, thereby constraining further improvements in performance. AlayaLaser relies on a large volume of in-memory data to achieve good performance, but its performance can be significantly impacted under memory constraints. In general, existing methods struggle to reduce disk access costs and achieve low-latency search with a lower memory footprint. We will give a detailed analysis in Section 3.

\textbf{Challenge.} We face the challenge of optimizing disk access cost to achieve low-latency search with reduced memory footprint. (1) Latency is primarily determined by the number of disk accesses. Effectively reducing the number of disk accesses is key to lowering search latency, which requires us to reduce disk I/O without erroneously filtering out valid data points. (2) Reducing the memory footprint is another critical factor. We need to achieve effective filtering with as little memory as possible. Simply shifting the burden by loading more data from the disk to memory to reduce disk access is not viable. (3) In order to achieve search latency comparable to or lower than HNSW, reducing the disk access cost is not sufficient. We also need to optimize in-memory process performance.

To address these challenges, we propose SkipDisk, a disk-memory hybrid ANN search for low-latency search. The disk storage holds the complete dataset, whereas the memory storage contains data such as neighbor information and other relevant data for candidate generation and filtering. SkipDisk identifies candidate points in memory. Then, it filters these candidates to reduce the number of disk accesses required for the final distance computation. In addition, it employs asynchronous processing to overlap in-memory search and disk access. Figure ~\ref{fig:intrmemlat} illustrates our performance and memory usage compared to HNSW on the BIGCODE dataset. SkipDisk achieves lower search latency than HNSW with a significantly reduced memory footprint. Specifically, we design SkipDisk from the following three perspectives.

\textbf{Reducing Disk Reads via In-Memory Filtering.} Our core strategy is to design effective filtering criteria using in-memory data to reduce disk accesses. Disk access is required to retrieve the full vector data of candidate points for the exact distance computation, but most retrieved points do not end up in the final result set. We use triangle inequality to compute lower bounds to filter candidate points. But in high-dimensional spaces, these bounds can be loose, leading to poor filtering. To address this, we maintain a separate pivot for each point, using low-precision values by preserving only a few high-significant bits of mantissa part of float for each dimension. This allows us to compute a tighter lower bound for each candidate point. As the cost of computing the distance between two points is much smaller than retrieving a data point from disk, we can afford the additional computational overhead of computing the distance between the query and a unique pivot point for each candidate to obtain a tight lower bound for filtering.We also design an estimation-based approach based on this lower bound to filter points. This allows us to filter out points when the lower bound falls below a threshold, and save memory footprint.

\textbf{Lowering Memory Footprint via Data Pruning.} The filtering capability is closely related to our goal of reducing memory footprint. Our core strategy involves a three-level approach to prune data in memory, thereby keeping the memory footprint low while still providing effective filtering to reduce disk accesses. The three layers are dimension-level, bit-level, and point-level. The first two levels aim to provide a good lower bound while minimizing memory usage. The third level allows us to incur a slight increase in disk accesses in exchange for further reducing memory usage. \textbf{(a)} At Dimension-level, we reduce memory by storing only the subset of dimensions that are most informative for distance computation. We identify and retain those that are crucial for deriving a tight lower bound or filtering metric, thus maintaining filtering power with far less data in memory. \textbf{(b)} At Bit-level, we apply the triangle inequality to construct a lower bound. Building on the dimension-level approach, we further reduce memory usage by retaining only the high significant bits of each dimension to construct pivot points. This allows us to further reduce memory usage while still providing a useful lower bound for filtering. \textbf{(c)} At Point-level, We allow some points to reside outside of memory and be accessed directly from disk, which can further reduce memory usage at the cost of some additional disk accesses. By carefully selecting which points to filter and which to access directly, we strike a balance that minimizes overall query latency while keeping memory footprint low.

\textbf{Overlap of I/O and Computation.} To further optimize search performance, we employ an asynchronous approach to overlap in-memory search and disk access. Inspired by PipeANN, we use asynchronous I/O and a poll-based I/O pipeline to fully leverage disk performance. We store neighbor nodes in memory, allowing the search for candidate points in memory to be decoupled from disk access. Whenever there are available slots in the I/O queue, we continuously issue disk access requests for candidate points that have passed the filtering stage. We also continuously retrieve data from completed disk reads and update the result set. Since we significantly reduce the number of disk accesses through effective in-memory filtering, I/O can be effectively hidden within the in-memory computation process. To further improve the in-memory graph traversal performance, we apply quantization to the PCA-reduced vectors used for in-memory search. This allows us to generate candidate points with a much lower distance computation cost while also reducing memory footprint. 

In summary, our contributions are as follows:
\begin{itemize}[leftmargin=20pt]
    \item We conducted an analysis of current methods, identifying disk access as the primary performance bottleneck, and limitation of existing filtering-based approach in effectively reducing disk accesses.
    \item We improve the lower bound of the triangle inequality by assigning a separate pivot for each point. We also design an estimation-based approach based on this lower bound for effective filtering. These filtering methods are combined with a three-level data pruning strategy to reduce memory footprint and disk access.
    \item We decouple in-memory search and disk access by storing neighbor nodes in memory, and design an asynchronous I/O strategy to overlap in-memory search and disk access to reduce search latency.
    \item SkipDisk provides different variants to meet various memory and latency requirements. Extensive experiments on multiple datasets show that SkipDisk significantly reduces memory usage while achieving latency comparable to or lower than HNSW.
\end{itemize}

\section{Preliminaries}
In this section, we introduce the preliminaries of approximate nearest neighbor (ANN) search, including the problem definition, basic concepts, and related theory to our proposed method.
\subsection{Problem Definition}
\begin{definition}[K Nearest Neighbor Search (KNN Search)]
  Given a dataset \( P \) consisting of \( n \) data points in a \( D \)-dimensional space, and a query point \( q \), the goal of ANN search is to find the top-\( K \) nearest neighbors of \( q \) in \( P \) based on a specified distance function \(\delta(p,q)\). The \(KNN\) search aims to return a set of points \( S \subseteq P \) such that \(|S| = K\) and for each point \( s \in S \), the distance \( \delta(s, q) \) is among the smallest distances from \( q \) to any point in \( P\setminus S \).
\end{definition}

Approximate nearest neighbor search allows for some approximation in the results to achieve faster query times, especially in high-dimensional space where exact search can be computationally expensive. We use recall as the primary metric to evaluate the accuracy of ANN search results, defined as the ratio of correctly retrieved nearest neighbors to the total number of true nearest neighbors. In practice, achieving high recall with low latency is the key objective of ANN search algorithms.

In this work, the distance we focus on is the Euclidean distance, defined as: \(\delta(x, y) = \sqrt{\sum_{i=1}^{D} (x_i - y_i)^2}\), where \(x\) and \(y\) are two points in the \(D\)-dimensional space. We use \(\delta^2(x, y)\) to denote the square Euclidean distance. We use \(LB(x,y)\) to denote the lower bound of two points.

In this paper, we focus on optimizing the latency of a query, which is the time required to retrieve search results for a given query point. Specifically, we address approximate nearest neighbor (ANN) search in a disk-memory hybrid setting. Our goal is to reduce memory footprint while achieving latency that matches or exceeds that of in-memory methods. We concentrate on float-type vector data, where each dimension is represented by 32 bits using the IEEE 754 standard. Float values are widely used in current AI applications for embedding representations, reflecting the practical scenarios encountered in real-world applications.

We use the term "filtering" to refer to the process of using in-memory data to reduce the need for accessing data on disk. For example, we can use an index or precomputed results in memory to filter out non-promising points, thereby reducing disk I/O operations. On the other hand, "pruning" refers to the strategy of deciding which part of data reside in memory and which is stored on disk to optimize memory utilization and performance. For instance, we may choose to keep frequently accessed data or critical dimensions in memory while storing less frequently accessed data or non-critical dimensions on disk.

\subsection{Search In DiskANN}
DiskANN is a classic disk-based ANN search method that serves as the foundation for subsequent disk-based ANN approaches. It reduces memory footprint by storing neighbor information and full-precision data on disk while keeping quantized data in memory. During the search, for each point selected from the search set, DiskANN fetches the point's neighbor information and original vector from disk. It then updates the search set using the quantized data in memory and recomputes distances with the full-precision data to re-rank the results. The algorithm~\ref{alg:diskann} summarizes this process. Although DiskANN reduces memory usage, high disk access costs lead to increased search latency. Subsequent optimizations, such as PipeANN, have achieved better performance.
\begin{algorithm}
  \caption{Search in DiskANN}
  \label{alg:diskann}
  \KwIn{Query \(q\), graph in disk \(\mathcal{G}\), quantized data \(\mathcal{Q}\), search parameters \(K\) and \(L\), search set \(C\), result set \(R\)}
  \KwOut{Top-\(K\) nearest neighbors of \(q\)}
  Initialize \(C\) with entry point\;
  \While{\(C\) has not been fully explored}{
    \(p \gets\) pop the top point from \(C\)\;
    get \(p\)'s neighbors from disk and update \(C\) using \(\mathcal{Q}\)\;
    get \(p\)'s original vector from disk and compute the distance to \(q\) to update \(R\)\;
  }
  \Return \(R\)\;
\end{algorithm}

\section{Analysis}
In this section, we give a detailed analysis on existing disk-based ANN methods. Our analysis has two parts. First, we show that disk access remains the dominant bottleneck, limiting the performance of existing disk-based methods. Second, we analyze the filtering power of Trim, showing that the trade-off between filtering power and accuracy leads to limited performance improvement.

\subsection{Disk Access Analysis}
Both DiskANN and PipeANN consist of two primary stages: (1) in-memory graph traversal and candidate expansion, and (2) disk reads for neighbor lists and candidate verification with full-precision vectors. In DiskANN, the search process alternates between in-memory processing and disk access, waiting for each disk access to complete before proceeding to the next in-memory processing step. In contrast, PipeANN uses asynchronous I/O to mask the latency of disk accesses. It overlaps in-memory processing with disk I/O, with a poll base IO pipeline to reduce the overall query latency.

We simplify the search process latency as:
\(
T_{\text{query}} = T_{\text{mem}} + T_{\text{io}} - T_{\text{overlap}},
\)
where \(T_{\text{mem}}\) is the in-memory processing time, \(T_{\text{io}}\) is the SSD access time, and \(T_{\text{overlap}}\) is the overlap due to asynchronous execution. For DiskANN, \(T_{\text{overlap}}\) is zero, so the query time is the sum of in-memory processing and disk access times. For PipeANN, \(T_{\text{overlap}} > 0\), but if disk access time is long, in-memory processing cannot fully mask the I/O latency, leading to high query latency.

\begin{figure}[ht]

    \centering
    
    \includegraphics[width=8.5cm,keepaspectratio]{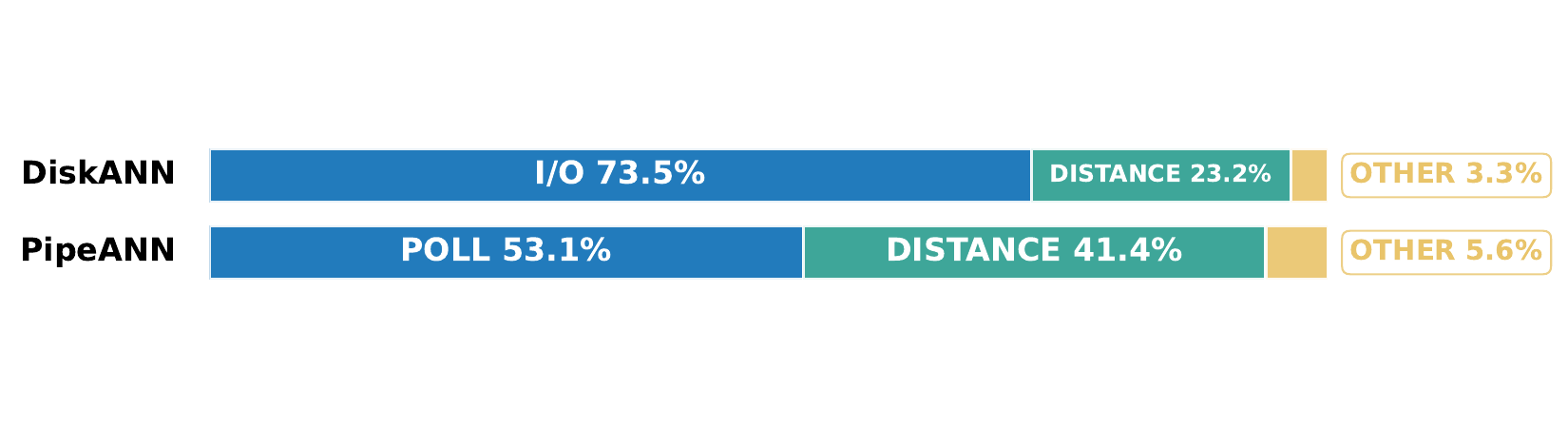}
  
    \caption{Time breakdown of DiskANN and PipeANN at Recall@99\% K=20 in MSMARCO dataset of 1536 dimensions. }
  
    \Description{}
    \label{fig:diskaccessanalysis}
\end{figure}

Figure ~\ref{fig:diskaccessanalysis} shows the query latency breakdown for both methods on an AWS EC2 i7i.4xlarge instance with instance's NVMe SSD. For DiskANN, disk access time dominates the total query time, indicating that disk I/O is the primary bottleneck. For PipeANN, the high polling time (waiting for disk access completion) indicates that in-memory processing only partially masks the I/O latency, with many disk I/O operations ongoing.

AlayaLaser~\cite{DBLP:journals/corr/abs-2602-23342} quantizes a vector every two sub-dimensions in the analysis, which results in a higher arithmetic intensity value that exceeds the \(I_{\text{max}}\) threshold in that paper, leading to a compute-bound conclusion. In practice, more sub-vectors are used, which results in a lower arithmetic intensity value that falls below the \(I_{\text{max}}\) threshold, leading to I/O bound conclusion. In the implementations of PipeANN and DiskANN, the maximum number of chunks is 256 and 512, respectively. In Trim, the default quantization is set to four dimensions per chunk. This configuration is sufficient to achieve good query performance while maintaining lower memory usage. In addition, in AlayaLaser's experiments, the comparison method used quantization encoding for each dimension, resulting in higher memory usage and high computational cost, and allowing AlayaLaser to load more data and achieve latency advantages. \textbf{Therefore, disk access remains the primary performance bottleneck in practice.} 

\subsection{Filtering Power Analysis}
Trim uses PQ (production quantization) codebooks as pivots to compute lower bounds based on the triangle inequality. It provides two modes: strict lower-bound mode and \(p\)-relaxed lower-bound mode, both of which aim to filter out points to reduce disk access. Gorgeous~\cite{yin2025gorgeous} introduces a parameter \(\sigma\) to limit the checking of the exact distances to only the top \(D_{r}\) points. By analyzing these methods, we demonstrate that \textbf{current filtering strategies face a significant trade-off between filtering out enough points to reduce disk access to achieve low latency and maintaining recall.}
\begin{figure}[ht]
    \centering
    \includegraphics[width=8.5cm,keepaspectratio]{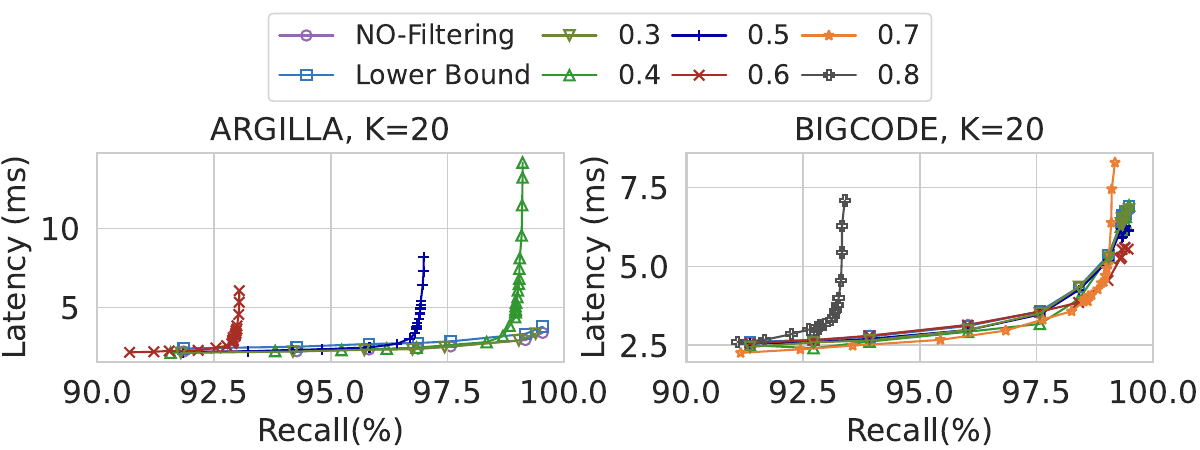}

    \caption{Latency of Trim on different parameters.}

    \label{fig:trimanalysis}
    \Description{}
\end{figure}

\begin{figure}[ht]
    \centering
    \includegraphics[width=7.5cm,keepaspectratio]{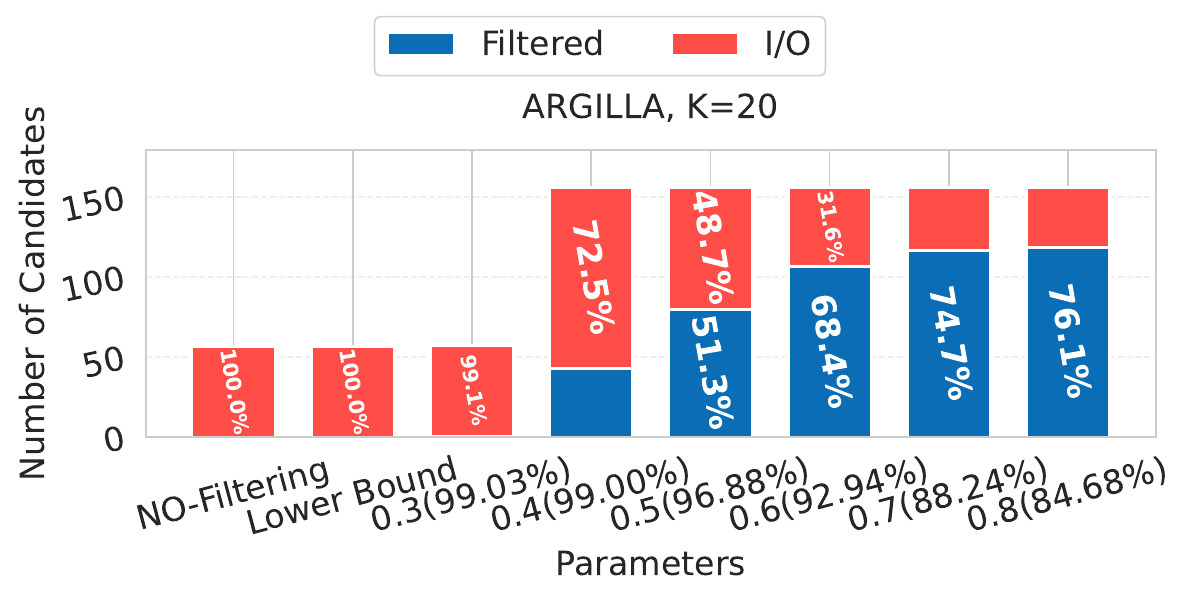}

    \caption{I/O of Trim on different parameters.}

    \label{fig:trimioanalysis}
    \Description{}
\end{figure}

\textbf{Strict lower-bound mode.}
Trim's strict lower-bound mode is safe but often weak. As shown in Figure~\ref{fig:trimanalysis}, enabling lower-bound filtering does not result in a significant performance improvement compared to the no-filtering case. The reason, as illustrated in Figure~\ref{fig:trimioanalysis}, is that very little data is actually filtered out. Consequently, the I/O bottleneck remains. Let \(\tau\) be the threshold for filtering, \(\delta(q, p)\) be the true distance between query \(q\) and candidate point \(p\), and \(\text{LB}(q, p)\) be the computed lower bound. We define the ratio \(r_{\text{ratio}} = \frac{\text{LB}(q, p)}{\delta(q, p)}\). To filter out point \(p\), we need \(r_{\text{ratio}} \cdot \delta(q, p) \geq \tau\). We observe that \(r_{\text{ratio}}\) is often around 0.6 on average, which is consistent with experimental results in Trim's paper. Due to the curse of dimensionality, many points have similar distances to the query. Moreover, the generated candidate points are often close to the query point, and their distances tend to be similar. This means that the lower bounds computed in this ratio often fall below the threshold \(\tau\), leading to many points that are actually far from the query not being filtered out. 

\textbf{\(p\)-relaxed lower-bound mode.}
To increase the filtering rate, Trim introduces a relaxed variant with a parameter \(\gamma\). This approach improves the quantity of filtering but also introduces false filtering, where relevant points are mistakenly discarded, directly reducing recall. As shown in Figure~\ref{fig:trimanalysis}, at lower \(\gamma\) values, latency does not decrease due to the limited number of filtered points. At higher \(\gamma\) values, performance does not improve because many relevant points are falsely filtered out. To achieve the same level of recall, more points need to be scanned to compensate for the relevant points that were incorrectly discarded, which means that there is no significant performance gain. This observation is also supported by Figure~\ref{fig:trimioanalysis}. At lower \(\gamma\) values, the proportion of filtered points is small. At higher \(\gamma\) values, although more points are filtered, the number of candidate points also increases significantly, which is consistent with the previous description.

\begin{figure}[ht]
\begin{subfigure}[t]{.22\textwidth}
  \centering
  \includegraphics[width=\textwidth]{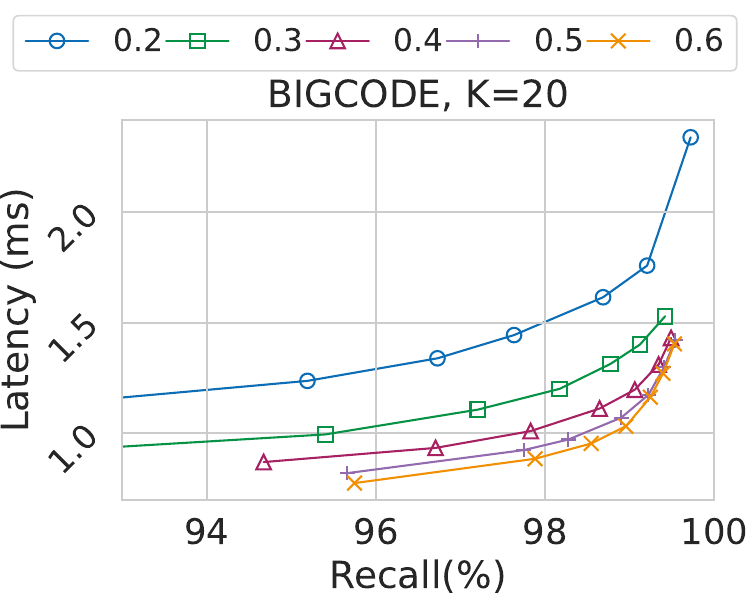}  
  \caption{Latency on different \(\sigma\)}
  \label{fig:latgorgsigma}
\end{subfigure}
\begin{subfigure}[t]{.22\textwidth}
  \centering
  \includegraphics[width=\textwidth]{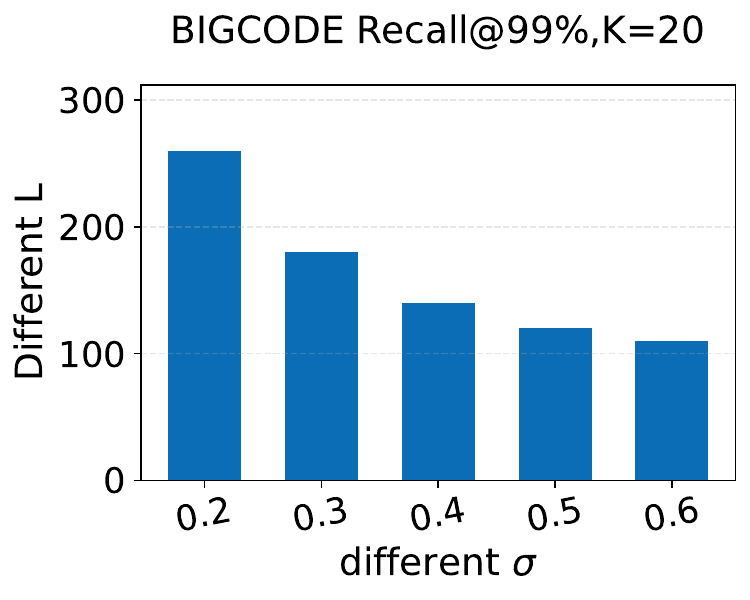} 
  
  \caption{L on different \(\sigma\)}
  \label{fig:lgorgsigma}
\end{subfigure}

\caption{Gorgeous analysis on different \(\sigma\).}

\label{fig:gorgeousanalysis}
\Description{}

\end{figure}

\textbf{Gorgeous.} Gorgeous~\cite{yin2025gorgeous} uses a parameter \(\sigma\) to control the number of candidate points that are checked after quantization-based search. Instead of examining all candidate points, it only checks the top \(D_{r}\) points determined by \(\sigma\). Reducing \(\sigma\) to minimize disk access can lead to discarding many points that would otherwise meet the result requirements. Figure~\ref{fig:latgorgsigma} illustrates the performance under different \(\sigma\) values. As \(\sigma\) decreases, latency gradually increases because more points need to be checked to compensate for the loss in precision, leading to a decline in overall performance. Figure~\ref{fig:lgorgsigma} further demonstrates this trade-off by showing that a larger \(L\) (checking more points during the search) is required to maintain a 99\% recall rate, highlighting the need to balance between reducing disk access and maintaining query accuracy.

\section{Basic Design}
In this section, we present the basic design of our proposed SSD-Memory search method, SkipDisk. We first provide an overview of the overall framework, and then delve into the details of our design. Build upon this foundation, we will introduce further optimization strategies in next section.
\subsection{Overview of SkipDisk}
SkipDisk is a disk-memory hybrid ANN search method that reduce end-to-end query latency under constrained memory. The core idea of SkipDisk is to leverage in-memory data to make informed decisions about whether a candidate data point is worth accessing on disk. If a point can be filtered using in-memory data, the disk access for that point can be skipped. SkipDisk follows a memory-for-I/O trade-off: a moderate increase in memory footprint is exchanged for a larger reduction in disk accesses and overall latency.

\begin{figure}[ht]
    \centering
    \includegraphics[width=8.5cm,keepaspectratio]{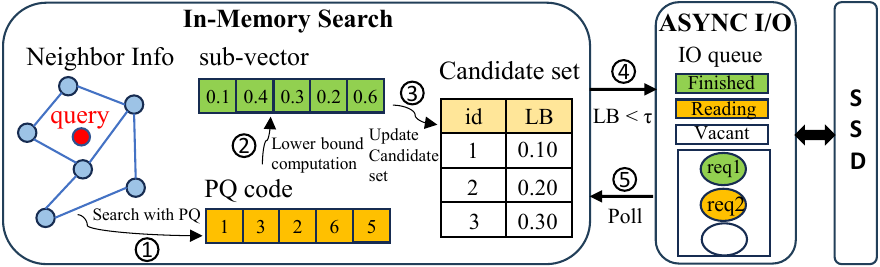}
   
    \caption{Workflow of SkipDisk. We follow these steps: \ding{172} Generate candidates using in-memory neighbor nodes and PQ data, \ding{173} Compute lower bounds using subspace vectors, \ding{174} Update the candidate set based on the lower bounds, \ding{175} Send I/O requests for points with lower bounds above the threshold, \ding{176} Retrieve the requested data asynchronously. }
    \label{fig:overallworkflow}
    \Description{}
\end{figure}

The overall framework of SkipDisk consists of three main components: (1) an in-memory search strategy, (2) a disk access filtering mechanism, and (3) an asynchronous I/O and computation overlap pipeline. Figure~\ref{fig:overallworkflow} illustrates the overall workflow of SkipDisk, showing how these three components interact during the search process. The in-memory search strategy is designed to efficiently identify promising candidate points based on the data stored in memory. The disk access filtering mechanism utilizes tight lower bounds and estimation techniques to filter out non-promising points before accessing disk. Finally, the asynchronous I/O and computation overlap pipeline allows concurrent execution of disk accesses and in-memory computations, further reducing query latency. The Algorithm~\ref{Alg:Basic} details the process. By carefully designing each component to balance memory footprint and disk access, SkipDisk can achieve latency comparable to or lower than HNSW, while using significantly less memory.  

\begin{algorithm}
  \caption{Basic SkipDisk}
  \label{Alg:Basic}
  \KwIn{query \(q\), search-set \(S\), search-set size \(L\), result-set \(R\), disk-candidate set \(C\), async-io-queue \(Q\), IO queue size \(B\)}
  \KwOut{Top-K nearest neighbors of \(q\)}
  \(\tilde{q} \gets \text{PCA\_transform}(q)\) \tcp*{PCA transformation}
  Initialize \(S\) with entry points using PQ\;
  Initialize \(threshold\) to infinity\;
  \While{\(S\) has elements to check}{
    \(p \gets\) closest unchecked point from \(S\)\;
    \(p_{\text{disk}} \gets\) get data point retrieved from \(Q\)\;
    \If{\(|Q| < B\)}{
        \If{\(C\) has elements}{
          \(p_{\text{cand}} \gets\) a point from \(C\) with smallest distance to \(q\)\;
        }
      \If {LB\((\tilde{q}, p_{\text{cand}})\) < threshold}{
        commit async read for \(p\) to \(Q\) and break\;
      }
    }
    check neighbors of \(p\) in memory and update \(S\) with distance using PQ\;
    compute \(LB(q, p)\) using PCA reduced data and update \(C\) with \(p\) and LB\;
    compute dist\((q, p_{\text{disk}})\), update \(R\) and threshold\;
  }
  process remaining points in \(C\) and \(Q\) \;
  \Return \(R\)\;
\end{algorithm}

\subsection{In-Memory Search}
The object of in-memory search is to generate candidate points that are potential final results within memory through graph traversal. These candidates are further filtered to determine which points need disk access for their original vector data. The index in graph-based methods can be divided into two parts: one is the graph connectivity information, which includes the neighbor list for each point, and the other is the vector information used for distance computations for each point.

\textbf{Neighbor Information.} During graph traversal, the algorithm accesses and expands a point's neighbor list iteratively. In DiskANN, neighbor lists are stored on disk, requiring disk I/O for each expansion. SkipDisk, however, keeps the full neighbor information in memory, allowing direct access without disk I/O. This decouples in-memory search from disk access, enabling efficient candidate generation and deferring costly disk reads until necessary.

Storing neighbor lists in memory is relatively inexpensive. For example, with 10 million 768 dimension points and 64 neighbors per point, the neighbor lists would require about 2GB of memory, which is small compared to the approximately 35GB memory footprint of the full HNSW index. Additionally, separating the storage of neighboring lists and vector data can reduce disk usage. In DiskANN and PipeANN, the data is stored in pages (e.g., 4KB per page). For a 1024-dimensional vector with its neighbor list, the total size can exceed 4KB, which requires two pages. This leads to underutilization of the second page and approximately double the disk usage. 

\textbf{Vector Information.} For the vector information used in memory for computing neighbor distances, we apply Principal Component Analysis (PCA) to each data point for dimensionality reduction, and our search is conducted based on the reduced-dimensional data inspired by SkipComputing~\cite{DBLP:journals/pacmmod/SongWY25}. To further reduce memory usage, we apply product quantization to these reduced-dimensional vectors. During the search process, we conduct the search on the data that has been dimensionally reduced and quantized to generate candidate points. After reducing disk access, this becomes the bottleneck in achieving a lower latency than HNSW. Using quantization data, we can improve in-memory search performance.

The process of generating candidate points is as follows: Our entire search strategy is designed based on the DiskANN, and we maintain a priority queue for storing candidate points. As we traverse the graph in memory, whenever a node's neighbor needs to be accessed, we add that node to the priority queue as a potential candidate. This node may require disk access to retrieve its original vector data.

\subsection{Dimensional-Level Pruning}
After generating candidate points in memory, we compute the lower bound of the distance to the query point for each candidate point. If the lower bound is greater than the current threshold, we can skip disk access for that point. Otherwise, we need to access the disk to retrieve the original vector data for distance computation.

The distance based on a subset of dimensions can serve as a lower bound for the distance between two vectors~\cite{DBLP:journals/pacmmod/SongWY25}, allowing us to avoid loading all vector dimensions into memory. This is the core idea of our Dimension-Level Pruning, which reduces memory usage by storing only a reduced representation of the vectors. Our Dimension-Level Pruning design performs PCA transformation on the data and stores only the first few dimensions in memory. As illustrated in Figure~\ref{fig:dimlevelpruning}, we reserve only part of the dimensions in memory and compute the lower bound for filtering.

\begin{figure}[ht]
    \centering
    \includegraphics[width=6.0cm,keepaspectratio]{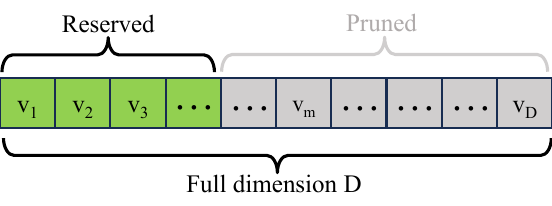}

    \caption{Dimensional-Level pruning for a vector.}

    \label{fig:dimlevelpruning}
\end{figure}
During the in-memory search process, for each generated candidate point, we compute the lower bound of the distance to the query point using the PCA subspace representation stored in memory. This lower bound serves as sorting criterion for the candidate points in the priority queue. We always check the point with the smallest lower bound distance first to determine if it requires disk access. We have the following property: during the checking process, if a point's lower bound distance exceeds the threshold, then the lower bound distances of all subsequent points will also exceed the threshold. Therefore, after all candidate points are generated, we can stop checking more points once we find the first point whose lower bound distance exceeds the threshold.

\subsection{Overlap of I/O and In-Memory Search}
Asynchronous data reading, such as in PipeANN, improves search performance by overlapping I/O with computation. We adopt a similar strategy to reduce query latency. Since we store neighbor node information in memory, we avoid SSD reads during graph traversal. Instead, we continuously generate candidate points in memory and use asynchronous I/O only when necessary points (with lower bounds below the threshold) need to be read from disk for actual distance calculations.

We maintain an async-io-queue for submitting asynchronous I/O requests. When the queue is not full and there are candidate points requiring disk access, we submit an I/O request to fetch the vector data for distance computation. Meanwhile, CPU threads continue to expand candidates in memory. We continuously check the queue for completed requests, compute distances with the retrieved data, and update the result set and threshold \(\tau\) accordingly. This enhances the filtering power for subsequent candidates. Points with lower bound distances greater than \(\tau\) skip disk access. This approach reduces disk reads and hides I/O latency behind computation, reducing overall search latency.

\subsection{Discussion}
Our basic design can effectively reduce latency, but the memory footprint is still a concern. Although we store PCA-reduced representations in memory, which significantly reduces memory usage, it can still be substantial, especially for large datasets. A key question is whether we can further reduce the memory footprint while still maintaining effective filtering performance, such as ensuring that the lower bound remains effective even with reduced memory usage. In the next section, we explore how to further reduce memory footprint while maintaining effective filtering performance to reduce disk access and achieve low latency search performance.

\section{Optimization Design}
In this section, we present optimization strategies to reduce the large memory footprint identified in the basic design of SkipDisk.

\subsection{Lower Bound with Bit-Level Pruning}
Our core object is to further reduce the memory footprint while maintaining effective filtering performance. Specifically, we aim to design a tight lower bound for filtering with reduced memory footprint. Our key idea is to perform bit-level reduction on the vector data for each dimension, which can further reduce memory usage. Based on this reduced representation, we can design a method to get a tighter lower bound for filtering using the triangle inequality. 

\begin{figure}[ht]
    \centering
    \includegraphics[width=8.5cm,keepaspectratio]{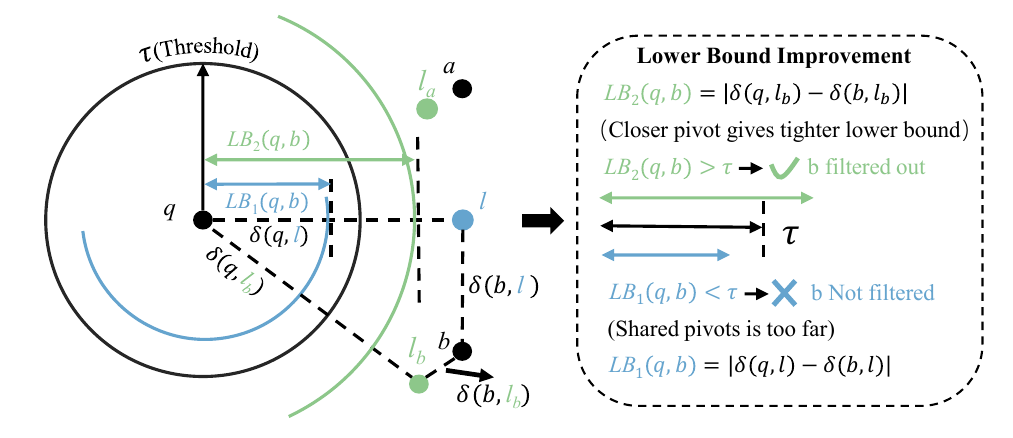}

    \caption{Lower bound. The distance of point \(b\) to \(q\) exceeds the threshold \(\tau\). The lower bound \(LB_{1}(q, b)\) calculated using the pivot point \(l\) designed for points \(a\) and \(b\) is below \(\tau\), so point \(b\) cannot be filtered out. By designing a closer pivot point \(l_b\) specifically for point \(b\), the new lower bound \(LB_{2}(q, b)\) exceeds \(\tau\), thus filtering out point \(b\).}
    \label{fig:methodlbana}
    \Description{}

\end{figure}

The triangle inequality is a widely used principle in index design~\cite{DBLP:journals/vldb/ZhuCGJ22}. For two points \(p\) and \(q\), with \(l\) as their pivot point, the lower bound can be calculated using the following formula: \( \delta(p,q) \geq |\delta(p,l) - \delta(q,l)| \). In conventional indexing methods, a classic approach is to select a single pivot point to represent a group of points. During index construction, the distances between each point in the group and the pivot point are calculated and stored. During the query phase, the distance between the query point and the pivot point \(l\) is calculated first. Then, the lower bound of the distance between the query point and each point in the group is computed using the triangle inequality. This allows us to filter out results that do not meet the requirements, while avoiding direct access to the original data, thereby improving query performance.

Due to the curse of dimensionality, it is difficult to find a pivot point that can effectively represent a group of points and provide a tight lower bound for filtering~\cite{DBLP:journals/pacmmod/SongZGYWWQ25}. The Trim method proposes using PQ-encoded points as pivots. The core idea is to cluster vectors in subspaces and assign a pivot point to each group, thereby significantly increasing the number of pivot points to provide better lower bounds. However, as analyzed in section 3.2, this approach still struggles to provide effective lower bounds for filtering in high-dimensional spaces. 
\begin{figure}[ht]

    \centering
    \includegraphics[width=8.5cm,keepaspectratio]{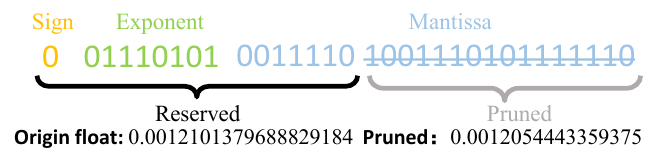}

    \caption{Bit-Level pruning for a float.}
    \label{fig:bitpruning}

\end{figure}

\textbf{Instead of finding a single pivot point to represent a group of points, we design a dedicated pivot point for each individual point.} Why does this approach help improve our performance? First, this approach aims to provide a tighter lower bound for filtering. From the triangle inequality, to obtain a tighter lower bound, we want \(\left|\delta(p,l) - \delta(q,l)\right|\) to be as close as possible to \(\delta(p,q)\). A feasible scenario is when \(\delta(q,l)\) is as large as possible and \(\delta(p,l)\) is as small as possible. If we assign a unique pivot point \(l\) to each point \(p\) such that \(\delta(p,l)\) is minimized (ideally close to zero), then \(l\) can effectively represent \(p\). Consequently, \(\delta(q,l)\) will provide a good approximation of the distance between \(q\) and \(p\), leading to a tighter lower bound. As shown in Figure ~\ref{fig:methodlbana}, we obtain a tighter lower bound by designing a closer pivot \(l_{b}\) specifically for point \(b\).

Second, the design of the lower bound is specifically for optimizing disk access. The cost of computing the distance between two points is much smaller than retrieving a data point from the disk. This characteristic provides us with a unique advantage in designing the lower bound computation method: We do not need to pay special attention to the efficiency of computing lower bound. As long as the cost of computing the lower bound is less than the cost of accessing the disk, we can achieve better performance. In this context, we are more concerned with the memory footprint of the lower bound computation method rather than its computational efficiency. We do not need to use a single pivot point to compute the lower bound for multiple points. The goal is to reduce memory usage while still providing effective filtering power. 

\textbf{How can we design this pivot point for each point?} We perform bit-level pruning on the vector data for each dimension, and the result values serve as the pivot points. Bit-level pruning means reducing the memory footprint of each floating-point number by retaining only the sign bit, exponent bits, and a subset of the most significant mantissa bits, as shown in Figure ~\ref{fig:bitpruning}. 

The embedding vectors are typically stored in 32-bit floating-point format (IEEE 754). The IEEE 754 floating-point representation~\cite{ieee754float} consists of a sign bit, exponent bits, and mantissa bits. Specifically, the 32-bit format is structured as follows: Sign bit (1 bit): Determines the sign of the number. Exponent (8 bits): Represents the exponent of the power of 2. Mantissa (23 bits): Encodes the fractional part of the number. The value of a floating-point number is calculated as:
\[ (-1)^{\text{sign}} \times 1.\text{mantissa} \times 2^{\text{(exponent - 127)}} \]
In the mantissa, each bit represents \(b_i \times 2^{-i}\), where \(b_i\) is the bit value (0 or 1) at position \(i\). The most significant bits (exponent and leading mantissa bits) capture the dominant information, while the less significant bits contribute less to the overall value. For example, MSB (bit position 0) contributes \(b_0 \times 2^{-1}\), while LSB (bit position 22) contributes \(b_{22} \times 2^{-23}\). By applying bit-level pruning to reduce the number of bits retained for each dimension, we can significantly reduce memory usage while still preserving the precision of the original data. This approach ensures that the value of \(\delta(p, l)\) becomes very small, close to 0, which in turn allows \(\delta(q, l)\) to better approximate \(\delta(p, q)\), providing a tighter lower bound.

\subsection{Lower Bound with BF16 and PCA}
In this section, we introduce the practical design of our lower bound computation in SkipDisk, which combines the BF16 format with PCA. This approach takes advantage of dimension-level pruning (via PCA) and bit-level pruning (via BF16) to reduce the memory footprint while still providing effective lower bounds for filtering. 
BF16 is a 16-bit floating-point format that retains the same exponent bits as standard 32-bit floats but reduces the mantissa to 7 bits. We chose BF16 for its balance between precision and memory usage, validated by its widespread use in deep learning. Additionally, using the BF16 format allows us to easily leverage SIMD instructions for efficient numerical computations. we store data in BF16 format. During distance calculations, we load multiple BF16 values into registers, pad them to FP32, and perform computations. Although the pivot points we used can be supported by different lengths of high significant bits, we use BF16 in the specific implementation in SkipDisk.

Using the BF16 format, we can reduce memory usage by half. However, this can still be substantial. To further reduce memory footprint, we first apply PCA for dimension-level pruning, and then apply BF16 format for bit-level pruning. We use this pruned representation as the pivot point to compute the lower bound using the triangle inequality. We have the following theorem to show that the distance computed using this approach is indeed a valid lower bound for the original distance between the two vectors:
\begin{theorem}
For a point \( p \) and a query point \( q \), let \(\tilde{p}\) be the PCA-reduced data of \( p \) and \(\tilde{q}\) be the PCA-reduced data of \( q \). Let \(\hat{p}\) be the BF16-pruned version of \(\tilde{p}\). We use \(\hat{p}\) as the pivot point for \(\tilde{p}\). \(|\delta(\tilde{q}, \hat{p}) - \delta(\tilde{p}, \hat{p})|\) is a valid lower bound for the original distance \(\delta(p,q)\). 
\end{theorem}
\begin{proof}
  The proof is straightforward. According to the triangle inequality, we have: \(|\delta(\tilde{q}, \hat{p}) - \delta(\tilde{p}, \hat{p})| \leq \delta(\tilde{q}, \tilde{p})\). Since \( \tilde{p} \) and \( \tilde{q} \) are PCA-reduced representations, the distance in the reduced space is a lower bound of the original distance. Formally, we have \(\delta(\tilde{q}, \tilde{p}) \leq \delta(q,p)\). Therefore, we can conclude that \(|\delta(\tilde{q}, \hat{p}) - \delta(\tilde{p}, \hat{p})| \leq \delta(q,p)\), which means that the computed value is indeed a valid lower bound for the original distance between \(p\) and \(q\).
\end{proof}

\subsection{Filtering with Estimation}
In this section, we explore the design of an estimation-based filtering method to further reduce disk access. The primary benefit of lower bound-based filtering is its correctness guarantee, ensuring that it will not erroneously filter out any points. However, this method relies on a tight lower bound, and to compute a lower bound that meets the filtering requirements, we need to retain more dimensions in memory. Can we still perform filtering when the lower bound is below the threshold? Our idea is to use the computed lower bound to estimate the distance between two points. We can compare this estimated value with a threshold for filtering. By doing so, we do not need to retain many dimensions of the data in memory. Although this introduces the risk of false filtering, it can lead to significant memory savings.

Various estimation-based DCO (Distance Comparison Operation) optimization methods have been proposed to improve the ANN search performance in memory. Among them, ADSampling~\cite{DBLP:journals/pacmmod/GaoL23} and DADE~\cite{DBLP:journals/pvldb/DengCZWZZ24} are based on computing partial distances to estimate the distance between two points. ADSampling relies on random projection, while DADE is based on PCA-rotated data. As we use PCA in our method, we design our estimation-based filtering method based on the DADE method. If we replace PCA with random projection, ADSampling can also be integrated. In DADE, the filtering condition is based on inequality \(dis' > \epsilon \cdot \tau\), where \(dis'\) is the distance computed over a subset of dimensions, \(\tau\) is the current threshold, and \(\epsilon\) is a parameter determined based on the number of sub-dimensions, their corresponding eigenvalues, and offline sampled data. In our method, we replace \(dis'\) with the lower bound computed using the triangle inequality based on PCA-reduced and BF16-pruned data. We then compare this estimated distance with \(\epsilon \cdot \tau\) to perform filtering. This method requires that the computed lower bound be sufficiently tight.  If the lower bound is significantly smaller than the actual distance in the subspace, it can lead to an underestimated value, leading to ineffective filtering. Our one-pivot-per-point strategy can provide a very tight lower bound, which allows our method to perform effective filtering.

\textbf{One step testing.} Unlike the estimation-based DCO optimization methods that rely on an incremental testing process to gradually increase the number of dimensions to reduce computational cost, our goal is to reduce disk access. We perform a single step testing based on the cached sub-dimensions in memory to perform filtering. If the point is not filtered after this testing, we directly access the disk to read the original vector data for a full distance computation and update the results accordingly. Because disk access is the most costly operation in this scenario, we want to minimize unnecessary disk accesses by using estimation-based filtering effectively, rather than performing multiple testing by requiring multiple disk accesses. And we also don't need to perform multiple testing in memory to optimize computational cost as the iterative process can lead to performance degradation~\cite{DBLP:journals/pacmmod/SongWY25}.

We have the following theorem to demonstrate that using our extended filter method does not introduce additional errors compared to the ADSampling and DADE methods. Specifically, if a point is filtered by our method, it will also be filtered by the ADSampling and DADE methods. Additionally, our method reduces the number of false negatives, meaning that points incorrectly filtered out by the ADSampling and DADE methods have a higher chance of being retained in our method in theory.

\begin{theorem}
let \(\tilde{p}\) be the PCA-reduced data of \( p \) and \(\tilde{q}\) be the PCA-reduced data of \( q \). Let \(\hat{p}\) be the BF16-pruned version of \(\tilde{p}\). Let \(LB(\tilde{q}, \hat{p}) = |\delta(\tilde{q}, \hat{p}) - \delta(\tilde{p}, \hat{p})|\) be the lower bound computed using our method, and let \(\delta(\tilde{q}, \tilde{p})\) be the distance computed in the PCA-reduced space used in the DADE method. Let \( \epsilon \) and \( \tau \) be parameters defined as in the DADE. Then:
\begin{itemize}
    \item \textbf{Correctness of Filtering.} If \( LB(\tilde{q}, \hat{p}) > \epsilon \cdot \tau \), then \( p \) will also be filtered out by the DADE method.
    \item \textbf{Reduction of False Negative.}  if \( p \) is filtered out by the DADE or ADSampling, it may not be filtered out by our method. This can be expressed as:   \( P(\delta(\tilde{q}, \tilde{p}) \geq \epsilon \cdot \tau) \geq P(LB(\tilde{q}, \hat{p}) \geq \epsilon \cdot \tau)\)
\end{itemize}
\end{theorem}

\begin{proof}
  The proof is straightforward, mainly based on the fact that \( LB(\tilde{q}, \hat{p}) \leq \delta(\tilde{q}, \tilde{p}) \). We omit the detailed proof.
\end{proof}

\subsection{Point-Level Pruning}
Point-level pruning means that we do not keep all points in memory, and allow some points to directly access the disk to obtain the original vector. As shown in ~\cite{DBLP:journals/pacmmod/SongZGYWWQ25}, during the search process, not all points have an equal chance of being accessed. Some points are accessed far more frequently than others. Therefore, to further reduce memory usage, we store the more frequently accessed points in memory. By adopting this approach, we make a trade-off between latency and memory usage, sacrificing some performance to achieve lower memory consumption. 

How do we select which points to keep in memory? Instead of designing a new strategy, we directly integrate the indegree-based strategy from AlayaLaser as our point selection strategy. We sort the points based on their indegree and select a series of points with the highest degree to be loaded into memory. The point-level pruning is built upon our dimension-level pruning and bit-level pruning, we do not need to store the full vector data in memory, which means that for the same memory budget, we can preserve more data points in memory.

\textbf{Search with Point-Level Pruning.} After obtaining candidate points through in-memory search, we need to determine which points require disk access to retrieve the original vector data for distance computation. For those points that are kept in memory, we can directly use the dimension-level pruning and bit-level pruning data to compute lower bounds for filtering. For those points that are pruned by point-level pruning, we directly set their lower bound values to 0. This means that no filtering is performed on these points, allowing them to be placed in the I/O queue for disk access.

\begin{figure*}
  \centering
  \includegraphics[width=\textwidth]{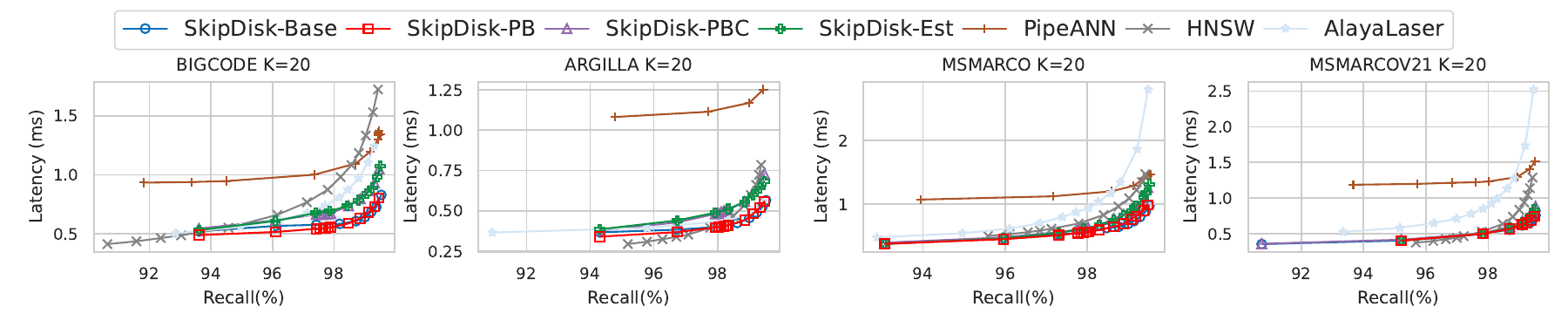}

  \caption{Overall Performance. Our method achieves lower latency compared to disk-based approaches. Compared to HNSW, our method achieves comparable or even lower latency while using significantly less memory.}

  \label{fig:overall perf}
  \Description{}
\end{figure*}

\section{Evaluation}

\textbf{Datasets.} We evaluate our method on several datasets, including ARGILLA~\cite{argilladataset}, BIGCODE~\cite{bigcodedataset}, MSMARCO~\cite{msmarcodataset}, and MSMARCOV21~\cite{msmarcov21dataset}. For MSMARCOV21, we randomly sample 50M points. These datasets are widely used in ANN search experiments and cover a range of common embedding dimensions and data scales. The datasets are available online such as on Hugging Face~\cite{huggingfacemainpage}. Table \ref{tab:datasets} provides an overview of the basic information for each dataset. We sample 10,000 data points from each dataset without replacement to form the query set and remain data points as the dataset for index.

\begin{table}
\caption{Dataset Statistics}

    \label{tab:datasets}
    \centering
    \begin{tabular}{c|c c c c}
    \hline
       Dataset  & \#Vectors & Dim & Data size & \#Queries \\
       \hline
       ARGILLA & 21,071,228 & 1024& 86.3 GB & 10000 \\
       BIGCODE & 10,404,628 & 768& 31.9 GB& 10000 \\
       MSMARCO & 8,831,823 & 1536& 54.3 GB & 10000 \\
       MSMARCOV21 & 50,000,000 & 1024& 204.8 GB & 10000 \\
    \hline
    \end{tabular}  

\end{table}

\textbf{Evaluation Metrics.} We evaluate the performance of our method using several key metrics, including latency, recall, and memory usage. Latency is measured as the time taken to return results for a query, while recall is calculated as the proportion of relevant items retrieved from the total relevant items in the dataset. We use latency-recall curves to illustrate the trade-off between these two metrics. A curve that is closer to the bottom-right corner indicates better performance.

\textbf{Our Method Variants.} We evaluate the performance of our method under different variants to demonstrate the effectiveness of each optimization strategy. The variants include:
\begin{itemize}[leftmargin=20pt]
    \item \textbf{SkipDisk-Base.} This variant implements the basic design of SkipDisk with only dimension-level pruning using PCA for lower bound computation, quantization on PCA-reduced data for in-memory search.
    \item \textbf{SkipDisk-PB.} This variant incorporates both dimension-level pruning and bit-level pruning using BF16 for lower bound computation and filtering.
    \item \textbf{SkipDisk-PBC.} This variant combines dimension-level pruning, bit-level pruning, and point-level pruning, where we load a subset of points in memory based on degree-based sampling.
    \item \textbf{SkipDisk-Est.} This variant extends the PB variant by integrating estimation-based filtering using the DADE method.
\end{itemize}

\textbf{Compared Methods.} We compare our method against several state-of-the-art ANN methods, including PipeANN, AlayaLaser and HNSW.
\begin{itemize}[leftmargin=20pt]
    \item \textbf{PipeANN}~\cite{DBLP:journals/corr/abs-2509-25487} is a method that optimizes DiskANN by overlapping asynchronous I/O with computation, significantly improving the search performance of DiskANN.
    \item \textbf{AlayaLaser}~\cite{DBLP:journals/corr/abs-2602-23342} is a method that focuses on using quantization data to design layout to optimize computation. And cache data in memory to reduce I/O.
    \item \textbf{HNSW} is a vanilla hierarchical navigable small world graph implemented in hnswlib~\cite{DBLP:journals/pvldb/LuKXI21} and widely supported by vector databases.
\end{itemize}

\textbf{Parameter Settings.} Disk-based graph is constructed based on DiskANN. We set the maximum degree of the graph to 64 and the construction list size to 200. For HNSW, we set M to 32 and efConstruction to 500. For the product quantization (PQ) encoding, we use one quantization code for every four dimensions. For AlayaLaser, we use PCA followed by RabitQ for quantization, which is the default setting. For PipeANN, we use the default beamwidth for I/O. For our method, the subspace dimension \(d_{\text{LB}}\) for lower bound computation is set to 512 for MSMARCO and 256 for other datasets. The dimension of the subspace \(d_{\text{DADE}}\) for DADE is set to 128, and the error parameter \(P_{\text{s}}\) used in DADE is set to 0.4, with \(\epsilon\) determined automatically based on sampled data. The subspace dimension \(d_{\text{PQ}}\) used for PQ is set to 384 for MSMARCOV21 and 256 for other datasets.  For point-level pruning, we keep half of the data points in memory. For the IO queue size, we set the maximum value to 32. By default, asynchronous IO is enabled for all methods. We use grid search to find the suitable parameters. For the subspace dimensions, we first set it to 80\% of the explained variance, and then we perform a grid search around this value.

Our C++ code was compiled on Ubuntu 24.04 using GCC 13.3.0 with the -O3 optimization flag. The tests were conducted on an AWS i7i.4xlarge instance. For the MSMARCOV21 dataset, we used an AWS i7i.8xlarge instance to perform the tests. We enabled AVX512 instructions to accelerate the distance computations, and all experiments were conducted in a single-threaded environment. For the asynchronous I/O part, we use an implementation based on io\_uring~\cite{iouringurl}. The SSD used in the experiments is the NVMe SSD provided with the instance.

\subsection{Overall Performance}
Figure ~\ref{fig:overall perf} illustrates the overall performance of our method compared to the baseline methods in different datasets. We have the following observations: (1) Our method outperforms HNSW in terms of latency while using less memory. In the MSMARCO dataset, our SkipDisk-PB method achieves 70\% of the latency of HNSW at a 99\% recall rate, using only about 22\% of the memory that HNSW requires. (2) Our SkipDisk-PB method achieves lower latency compared to other disk-based methods. For example, on the BIGCODE dataset, with similar memory usage, our method achieves 63\% of the latency of AlayaLaser at a 99\% recall rate. (3) Our low-memory variants, SkipDisk-PBC and SkipDisk-Est, also achieve lower latency with further reduced memory usage. For example, on the MSMARCO dataset, our SkipDisk-Est method achieves 82\% of the latency of HNSW at a 99\% recall rate, using only about 10\% of the memory required by HNSW. (4) The performance of AlayaLaser is inferior to that of HNSW mainly because we limit its memory footprint to the same level as our SkipDisk-PB method. This restriction significantly reduces its ability to reduce disk access, leading to poorer performance. (5) The SkipDisk-PBC method, which employs point-level pruning, further reduces memory footprint but results in slightly higher latency compared to the SkipDisk-PB method. (6) Compared to SkipDisk-Base, our SkipDisk-PB method achieves nearly the same latency, but with significantly lower memory usage. This indicates that our bit-level pruning technique can substantially reduce memory consumption without increasing latency.
\begin{table}
\caption{Memory Footprint (GB)}

    \label{tab:ourmemfootprint}
    \centering
    \begin{tabular}{c|c c c c }
    \hline
       Method  & ARGILLA & BIGCODE     & MSMARCO      & MSV21  \\
       \hline
       Base    & 38.08 & 18.91  &   20.76   &91.92  \\
       PB      & 22.72 & 11.33  &  12.17   &54.75   \\
       PBC     & 15.41 & 7.78  &   8.18   &38.18    \\
       Est     & 12.43 & 6.24  & 5.70   & 30.31    \\
       HNSW    & 89.73 & 34.12 & 54.95   & 213.13   \\
       PipeANN & 5.49 &  2.23  &   3.53    & 12.78  \\
       AlayaLaser & 20.51 & 11.67  & 13.29  & 69.72 \\
    \hline
    \end{tabular}  

\end{table}

\subsection{Memory footprint}
Table ~\ref{tab:ourmemfootprint} presents the memory usage of our method across different datasets together with compared methos. MSV21 is short for MSMARCOV21. By enabling bit-level pruning based on BF16 and dimension-level pruning based on PCA, our method significantly reduces memory consumption. Furthermore, the introduction of point-level pruning further decreases memory usage while maintaining high performance. The use of an estimation method further reduces memory requirements.

Compared to HNSW, our method demonstrates a significant advantage in memory usage. On the MSMARCO dataset, our method consumes only 10.38\% of the memory used by HNSW. For the MSMARCOV21 dataset, our method uses only 14.22\% of the memory required by HNSW. Consequently, our method can achieve comparable latency performance to HNSW with just 30.31GB of memory, as opposed to the more than 200GB required by HNSW.

\subsection{Points number in Point Level Pruning}
In this section, we present the latency of our SkipDisk-PBC variant under different numbers of points loaded into memory when point-level pruning is enabled. Loading more points into memory reduces the need for disk access but also increases memory usage. Figure ~\ref{fig:datalevel} illustrates the performance of our method when a fixed dimension-level subspace of 384 dimensions is used, along with BF16-based bit-level pruning, and varying numbers of points are loaded into memory. As more points are loaded into memory, less data needs to be accessed from the disk, which reduces latency.

\begin{figure}[ht]
\begin{subfigure}[t]{.22\textwidth}
  \centering
  \includegraphics[width=\textwidth]{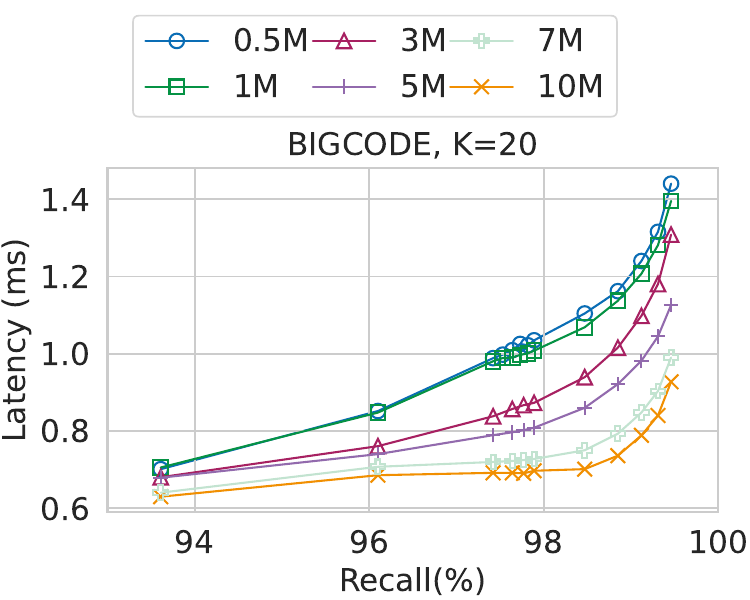}  
  \caption{Search latency}
  \label{fig:latencydatalevel}
\end{subfigure}
\begin{subfigure}[t]{.22\textwidth}
  \centering
  \includegraphics[width=\textwidth]{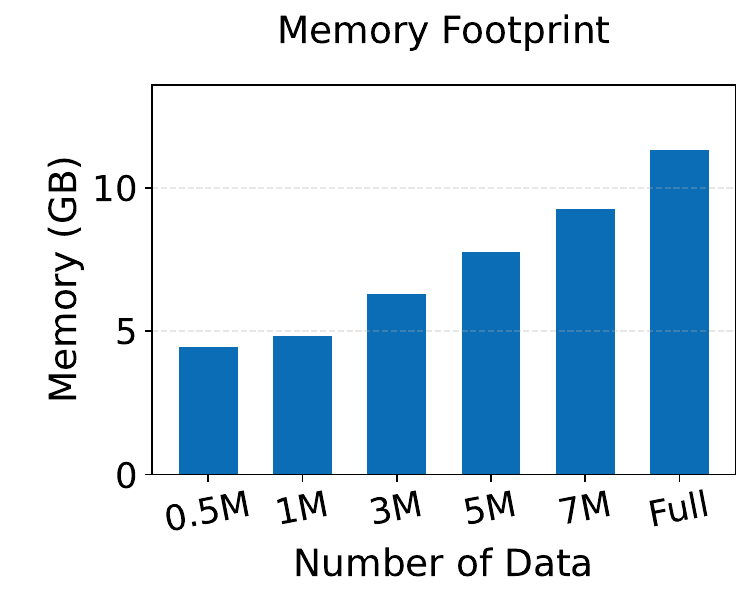} 
  
  \caption{Memory footprint}
  \label{fig:memdatalevel}
\end{subfigure}

\caption{Illustration of search latency and memory footprint on different data level pruning of SkipDisk-PBC.}

\label{fig:datalevel}
\Description{}
\end{figure}

\subsection{Latency under different Sub-dimensions}
The choice of different sub-dimensions significantly impacts search latency. Quantization and filtering can be performed on various sub-dimensions, and these choices affect the overall search performance. For quantization, if the data is overly reduced, it can lead to a loss of distance information, necessitating the traversal of more points to meet recall requirements, thereby increasing search latency. Conversely, if the sub-dimension is too large, the computational overhead increases. For lower bound based filtering, if the sub-dimension is too low, the filtering capability is weakened, leading to more disk accesses to meet recall requirements, which also increases search latency. We show the analysis results on SkipDisk-PB.

\begin{figure}[ht]
\begin{subfigure}[t]{.22\textwidth}
  \centering
  \includegraphics[width=\textwidth]{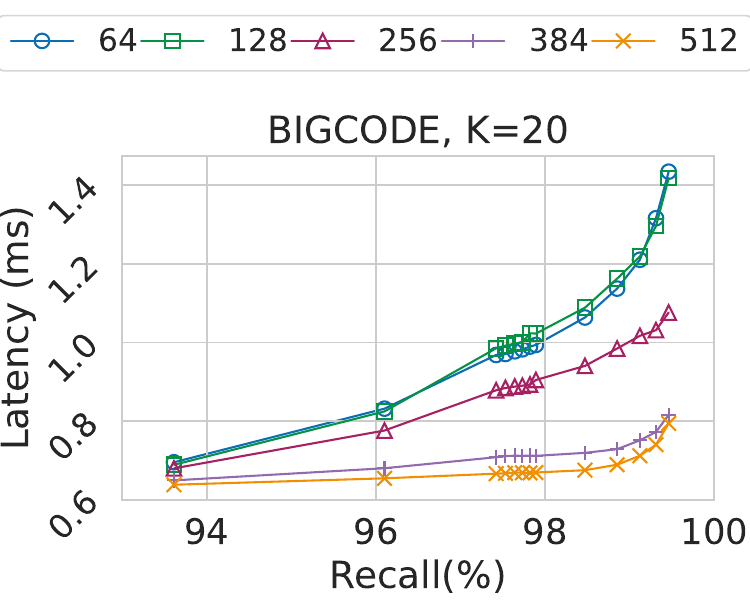}  
  \caption{Sub-dim for lower bound}
  \label{fig:latencysubdimlb}
\end{subfigure}
\begin{subfigure}[t]{.22\textwidth}
  \centering
  \includegraphics[width=\textwidth]{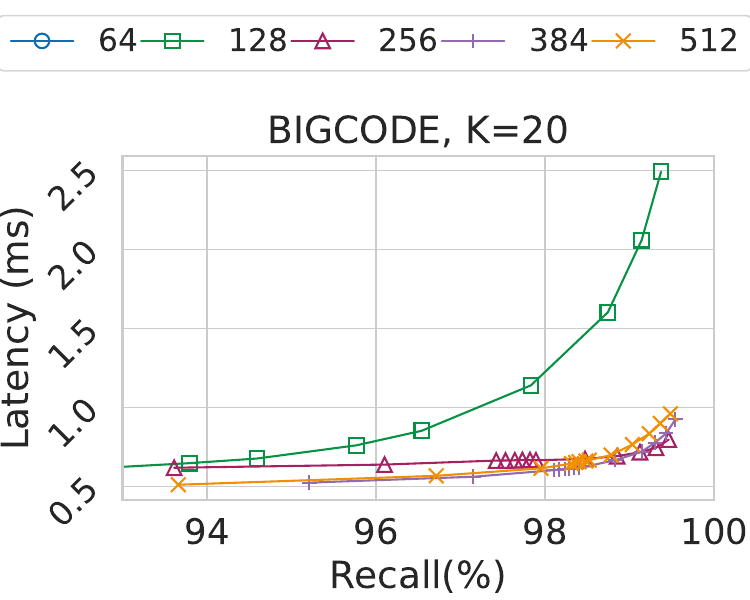} 
  
  \caption{Sub-dim for PQ}
  \label{fig:latencysubdimpq}
\end{subfigure}

\caption{Latency on different subdim of SkipDisk-PB.}

\label{fig:subdimonlbpq}
\Description{}

\end{figure}

Figure ~\ref{fig:latencysubdimlb} illustrates the search latency across four different sub-dimensions used for filtering. As the sub-dimension increases, the search latency decreases due to improved filtering performance, which reduces the number of disk accesses. However, beyond a certain point, further increasing the sub-dimension leads to diminishing returns in latency reduction, as the improvement in filtering performance becomes marginal. Figure \ref{fig:latencysubdimpq} illustrates the search latency across four different sub-dimensions used for product quantization (PQ), with the filtering dimension fixed. We observe that as the sub-dimensions increases, the search latency initially decreases, reaching an optimal point. Beyond this point, further increasing the sub-dimension leads to diminishing returns in latency reduction.

\subsection{I/O Reduction on Lower Bound Filtering}
In this section, we explore how our SkipDisk-PB variant reduces disk accesses under different sub-dimensions. Higher sub-dimensions lead to stronger filtering capabilities, thereby reducing the number of disk accesses, but they also result in increased memory usage. Figure~\ref{fig:subdimIOmemorlb} illustrates the changes in disk access and memory usage for different sub-dimensions in the BIGCODE dataset under 99\% recall rate. As the sub-dimension increases, there is a significant reduction in disk accesses, but this comes with a corresponding increase in memory usage. Therefore, in practical applications, it is essential to select an appropriate sub-dimension based on specific performance requirements and resource constraints to achieve optimal performance.

\begin{figure}[ht]
\begin{subfigure}[t]{.21\textwidth}
  \centering
  \includegraphics[width=\textwidth]{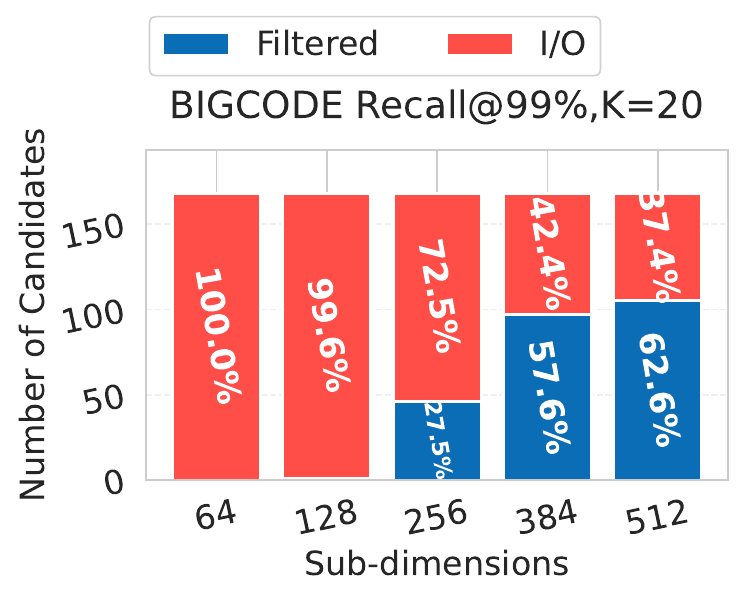}  

  \caption{I/O reduction}
  \label{fig:ioreductionsubdim}
\end{subfigure}
\begin{subfigure}[t]{.21\textwidth}
  \centering
  \includegraphics[width=\textwidth]{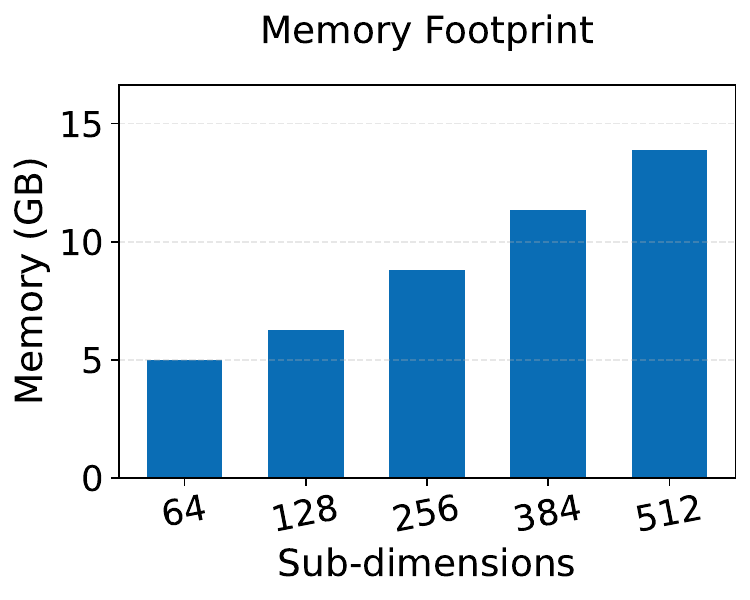} 

  \caption{Memory footprint}
  \label{fig:iomemsubdim}
\end{subfigure}

\caption{Disk I/O reduction and memory footprint on different dimension level pruning of SkipDisk-PB.}

\label{fig:subdimIOmemorlb}
\Description{}
\end{figure}

\subsection{Lower bound Tightness}
The tightness of the lower bound is crucial for latency, as it directly impacts the effectiveness of filtering. Here, we present the tightness of the lower bound computed by our method in Figure ~\ref{fig:lowerboudanalysis}. We compute the average lower bound values using BF16 for bit-level pruning as pivot across different datasets. We also include the average lower bound values from the Trim method for comparison. The results show that the lower bound values computed by our method are almost equivalent to the true distance values, while the lower bound values from the Trim method are significantly smaller than the true distance values. This indicates that our method provides a relatively tight lower bound to reduce disk accesses.

\begin{figure}
  \centering
  \includegraphics[width=8.5cm,keepaspectratio]{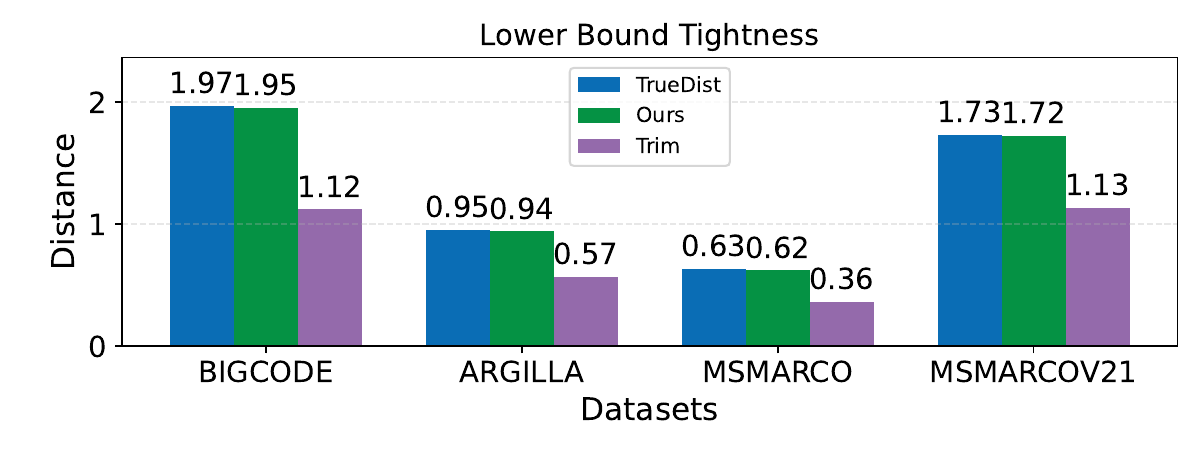}

  \caption{Tightness of lower bound.}
  \label{fig:lowerboudanalysis}

\end{figure}

\subsection{I/O Reduction on Estimation based Filtering}
The effectiveness of the estimation-based method in reducing disk accesses varies with different parameter settings. We present the reduction in disk accesses achieved by the estimation-based method under various parameter configurations in Figure ~\ref{fig:ioreductionest}. As the parameters increase, the filtering capability generally becomes stronger. However, larger parameters can also lead to the exclusion of some valid points, necessitating a more extensive search to meet recall requirements. For example, at a parameter value of 0.6, the result shows that a higher number of candidate points is needed to achieve a 99\% recall rate, even though the filtering capability is enhanced, resulting in a higher proportion of filtered points. In some parameter settings, the estimation-based method can significantly reduce disk accesses, while in others, the filtering performance may be less effective. Therefore, finding the optimal parameter settings is crucial for achieving the best filtering performance, which in turn enhances the overall search performance.
\begin{figure}
  \centering
  \includegraphics[width=6.5cm,keepaspectratio]{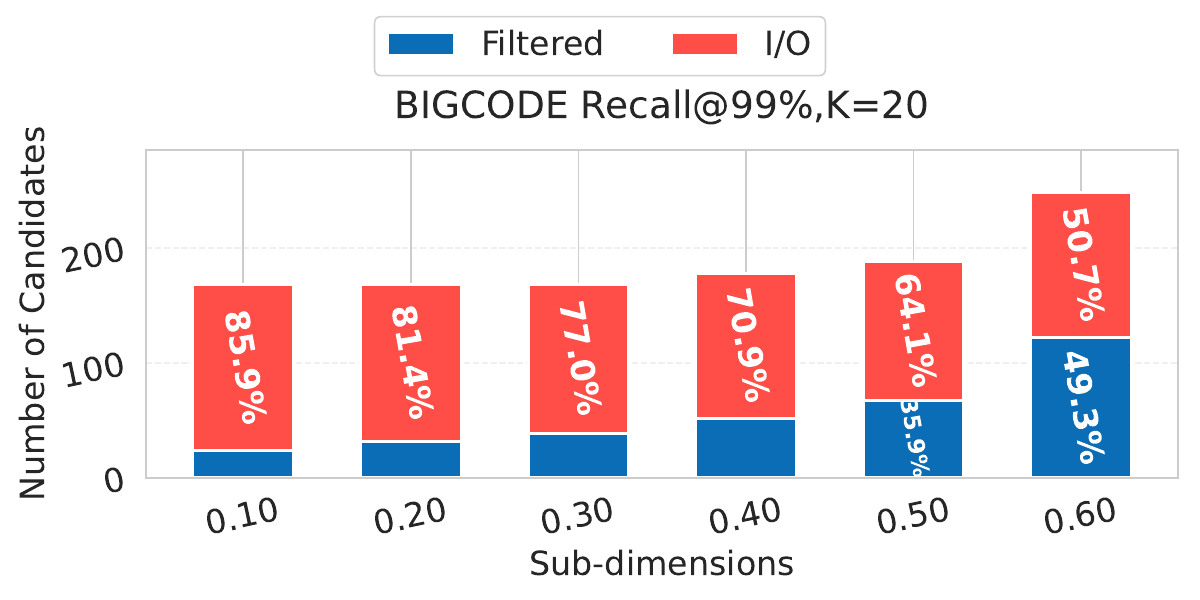}

  \caption{I/O reduction on estimation based filtering.}

  \label{fig:ioreductionest}
  \Description{Eight graph of different dataset to show our method's performance is better to compared methods.}
\end{figure}

\subsection{Performance on different top-\(K\)}
Figure \ref{fig:varytopk} illustrates the latency of our method under different top-\(K\) values. As the top-\(K\) value increases, the search becomes more challenging because more relevant items need to be returned to meet recall requirements, leading to more disk accesses and increased search latency. At top-\(K\)=100, our method still achieves lower latency than HNSW. However, as \(K\) increases to 500, our latency exceeds that of HNSW. Despite this, our memory usage is significantly lower than HNSW, and the increase in latency is relatively modest, with HNSW's latency being approximately 78\% of ours.
\begin{figure}[ht]
    \centering
    \includegraphics[width=7.5cm,keepaspectratio]{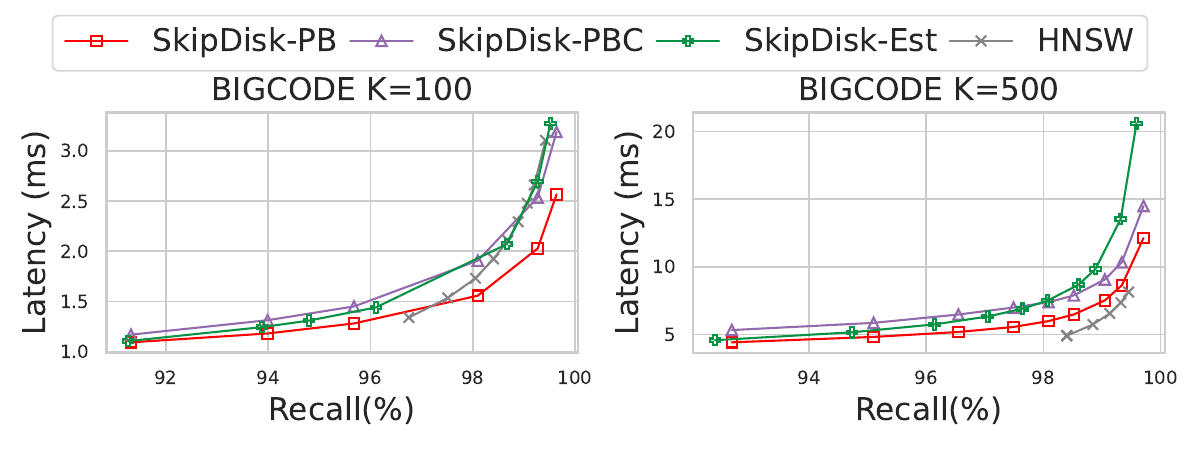}
    
    \caption{Latency on varying topK.}

    \label{fig:varytopk}
    \Description{}
\end{figure}

\subsection{Compared to PipeANN with Same Memory Budget}
PipeANN stores only quantized data in memory, resulting in relatively low memory consumption. In contrast, our method requires more memory. To compare the latency performance, we adjust PipeANN's memory usage to be approximately equivalent to that of our SkipDisk-PB variant by loading additional data into memory. The data loaded is selected based on a point-level pruning strategy. Figure ~\ref{fig:pipecache} presents the comparison results on the BIGCODE and ARGILLA datasets. The results show that, under the approximately same memory usage, our method significantly outperforms PipeANN in terms of latency performance on both datasets. This indicates that our method is more effective in utilizing memory resources and provides better overall performance. Furthermore, PipeANN achieves latency comparable to HNSW, demonstrating the excellence of PipeANN's design while also highlighting the achievements of our method.

\begin{figure}[ht]
    \centering
    \includegraphics[width=8.5cm,keepaspectratio]{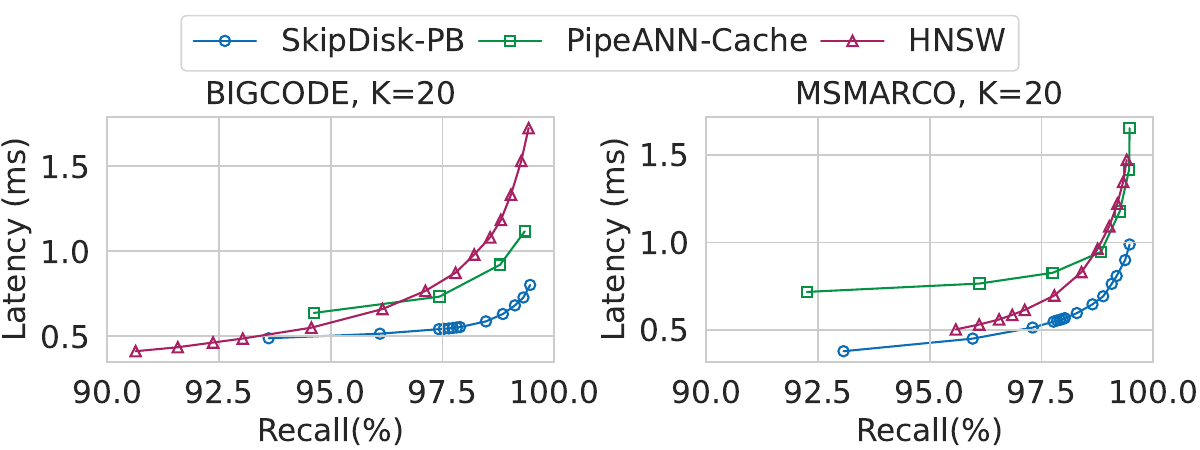}

    \caption{Latency of PipeANN with caching data points in memory.}

    \label{fig:pipecache}
    \Description{}
\end{figure}

\begin{figure}[ht]
    \centering
    \includegraphics[width=8.5cm,keepaspectratio]{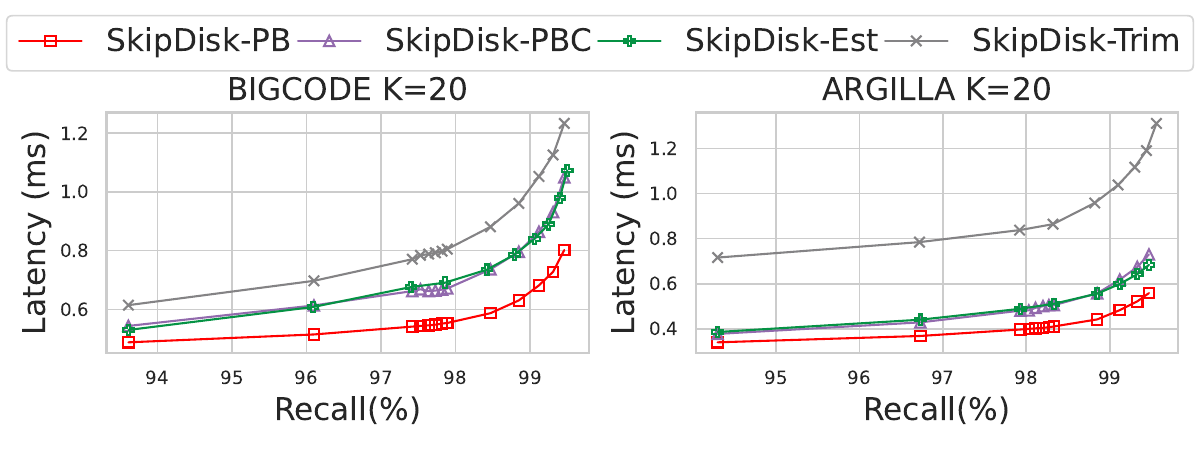}

    \caption{Latency of SkipDisk-Trim.}

    \label{fig:skipdisktrim}
    \Description{}

\end{figure}

\subsection{Extend with Trim}
Given that the Trim method can filter candidate points, we extend our approach to incorporate Trim's p-relax lower bound method for filtering and evaluate its performance. We compare this extension with our other variants on the bigcode and argilla datasets, setting the gamma parameter to 0.6 and 0.35, respectively. Figure ~\ref{fig:skipdisktrim} illustrates the performance comparison on these datasets. The results show that using our designed filtering method significantly outperforms the Trim method in both datasets. Specifically, after incorporating Trim, the memory usage on the BIGCODE and ARGILLA datasets is 5.90 GB and 12.93 GB, respectively, which is comparable to the memory footprint of our SkipDisk-Est method. However, the performance of the SkipDisk-Est method is notably superior. This comparison highlights the effectiveness of our triangle based filtering method in achieving better performance while maintaining a similar memory footprint.

\subsection{99.9th Percentile Latency}
In our previous evaluations, we focused primarily on the average latency performance. However, in practical applications, the  99.9th percentile latency performance is also a critical metric, as it reflects the system's behavior under extreme conditions. Figure ~\ref{fig:999latency} illustrates the 99.9th percentile latency performance of our SkipDisk-PB variant and compares it with that of HNSW. The results show that, even at the 99.9th percentile latency level, our method maintains comparable or lower latency performance than HNSW in BIGCODE dataset. This indicates that our SkipDisk-PB variant not only excels in average latency but also demonstrates robust performance under extreme conditions. These findings further validate the effectiveness and robustness of our approach.
\begin{figure}[ht]
    \centering
    \includegraphics[width=8.5cm,keepaspectratio]{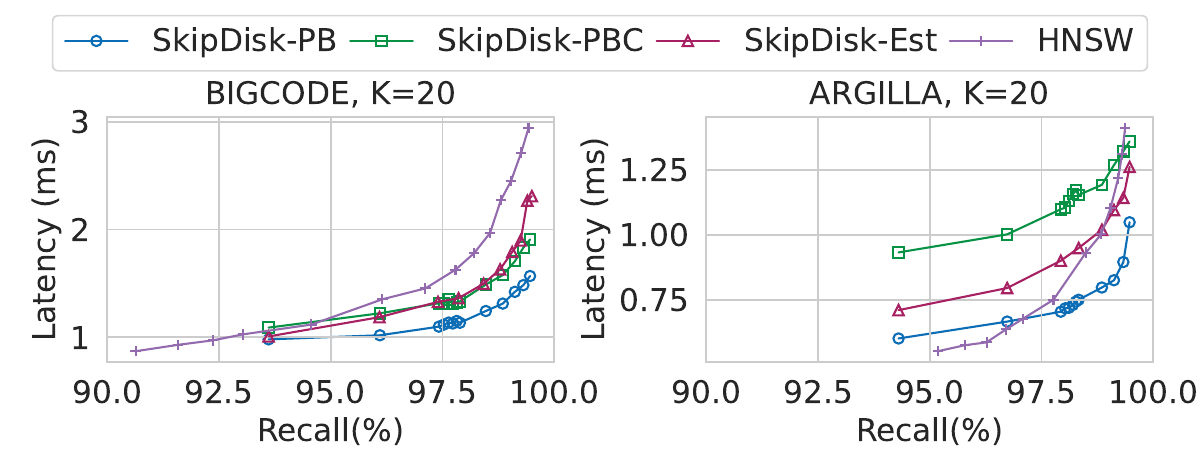}

    \caption{99.9th Percentile Latency on SkipDisk and HNSW.}
    \label{fig:999latency}

\end{figure}

\begin{figure}[ht]
    \centering
    \includegraphics[width=8.5cm,keepaspectratio]{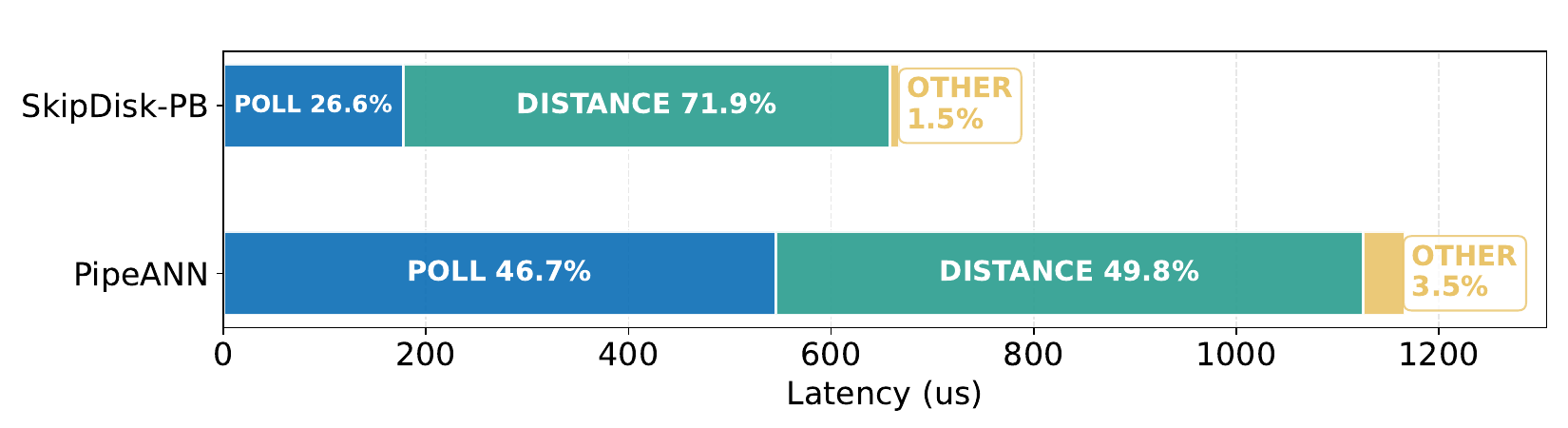}

    \caption{Time distribution of PipeANN and SkipDisk-PB an Recall@99 on K=20 in BIGCODE dataset.}

    \label{fig:timeanalysis}
    \Description{}
\end{figure}

\subsection{Latency Breakdown for SkipDisk-PB}
We perform a detailed latency breakdown analysis for the SkipDisk-PB variant to better understand the performance changes. Figure ~\ref{fig:timeanalysis} presents the results on the BIGCODE dataset, compared with PipeANN. We observe that previously, a significant portion of the overhead in PipeANN was due to polling, which involves continuously checking for I/O completion. After our optimizations, this portion has been significantly reduced as we reduce the number of I/O operations, with most of the I/O operations can be hidden within the candidate generation process.
\section{Related Work}
\textbf{In-Memory Graph-Based ANN methods.} Classic methods for ANN search include tree-based~\cite{DBLP:conf/icml/BeygelzimerKL06,DBLP:journals/cacm/Bentley75,DBLP:conf/vldb/CiacciaPZ97,DBLP:conf/cvpr/Silpa-AnanH08}, hashing-based~\cite{andoni2008near,andoni2015optimal,DBLP:conf/sigmod/GanFFN12,DBLP:journals/pvldb/HuangFZFN15,DBLP:conf/compgeom/DatarIIM04,DBLP:conf/sigmod/TaoYSK09,DBLP:journals/pvldb/SunWQZL14,DBLP:journals/vldb/ZhengZWNLJ22,DBLP:conf/sigmod/ZhengGTW16,DBLP:conf/iclr/DongIRW20,DBLP:journals/pieee/WangLKC16}, quantization-based~\cite{DBLP:journals/pvldb/AndreKS15,DBLP:conf/cvpr/GeHK013,DBLP:conf/cvpr/GeHK013,DBLP:conf/icassp/JegouTDA11,DBLP:conf/cvpr/KalantidisA14,10.1145/3654970}, and graph-based methods~\cite{DBLP:conf/www/DongCL11,fu2021high,DBLP:journals/pvldb/FuXWC19,DBLP:journals/pvldb/LuKXI21,malkov2018efficient,DBLP:conf/ijcai/HajebiASZ11,DBLP:conf/mm/WangL12,10.1145/3709693,DBLP:journals/pvldb/WangXY021}. Graph-based methods have gained significant attention due to their efficient search performance, with HNSW~\cite{malkov2018efficient} being a widely supported in vector databases\cite{DBLP:conf/sigmod/WangYGJXLWGLXYY21,DBLP:journals/pvldb/GuoLXYYLCXLLCQW22,DBLP:journals/pvldb/ChenJZPWWHSWW24,DBLP:journals/vldb/PanWL24}. Many researchers have optimized graph-based methods from various perspectives, such as edge connection strategies, to improve performance. Examples include FANNG~\cite{DBLP:conf/cvpr/HarwoodD16}, NSG~\cite{DBLP:journals/pvldb/FuXWC19}, Vamana~\cite{jayaram2019diskann}, \(\tau\)-MG~\cite{DBLP:journals/pacmmod/PengCCYX23}, \(\alpha\)-CNG~\cite{DBLP:journals/corr/abs-2510-05975}, and LMG~\cite{DBLP:journals/tods/XieYL25}. Our SkipDisk method directly uses the Vamana graph constructed by DiskANN for the in-memory part. In the future, we can explore the integration of other graph-based methods to further enhance in-memory search performance.

\textbf{Disk-Based ANN Search.} Disk-based ANN search has gained attention for handling large datasets. DiskANN~\cite{jayaram2019diskann} is a representative method. SPANN~\cite{DBLP:conf/nips/ChenZWLLLYW21}, an inverted index-based method, is easy to maintain~\cite{DBLP:conf/sosp/XuLLXCZLYYYCY23} but suffers from high latency due to the large amount of disk data access required~\cite{DBLP:journals/corr/abs-2602-23342}. Starling~\cite{DBLP:journals/pacmmod/WangXYWPKGXGX24} and PageANN~\cite{DBLP:journals/corr/abs-2509-25487} face limitations with higher-dimensional vectors due to disk page constraints. PipeANN~\cite{DBLP:conf/osdi/GuoL25} employs a asynchronous I/O and poll based pipeline to optimize disk access. OrchANN~\cite{DBLP:journals/corr/abs-2512-22838} employs a route–access–verify pipeline to optimize disk performance. And the center of the cluster for the pivot in the inferior bound for the lower bound as shown in Trim~\cite{DBLP:journals/pacmmod/SongZGYWWQ25}. VeloANN~\cite{DBLP:journals/corr/abs-2602-22805} targets improving system throughput by considering multiple query coroutines to progress cooperatively, while our method focuses on optimizing single-query latency. Trim~\cite{DBLP:journals/pacmmod/SongZGYWWQ25} reduce the disk access number using a lower bound based on the triangle inequality, but the loose bound may result in a limited filtering effect in some cases‌. Gorgeous~\cite{yin2025gorgeous}  reduces disk access by only examining the top-\(D_r\) candidates sorted by distances in production quantization data. Fusion-ANNS~\cite{DBLP:conf/fast/TianLTXDL0ZZ025} and GUSTANN~\cite{DBLP:journals/pacmmod/JiangGXSL25} leverages CPU-GPU co-design to improve throughput. AlayaLaser~\cite{DBLP:journals/corr/abs-2602-23342} optimizes data layout for computational efficiency and caches large amounts of data in memory, achieving good performance but with a high memory footprint. Gorgeous~\cite{DBLP:journals/corr/abs-2508-15290} redesigns the data layout to achieve better performance than Starling. Our SkipDisk is the first method to achieve better single-query latency than HNSW while using significantly less memory. In the future, we can explore the integration of other disk-based methods to further improve overall performance.

\textbf{Distance Computation Optimization.} Distance computation is a critical component of ANN search, and optimizing it can significantly enhance search performance~\cite{DBLP:journals/pacmmod/SongWY25,DBLP:journals/pacmmod/GaoL23}. Quantization-based approaches, such as RabitQ~\cite{10.1145/3654970}, SAQ~\cite{DBLP:journals/pacmmod/LiDYZC25}, MRQ~\cite{yang2024quantization} have been proposed to compress the size of the data and improve the distance computation. DCO (Distance Comparison Operation) based methods, such as ADSampling~\cite{DBLP:journals/pacmmod/GaoL23}, FINGER~\cite{DBLP:conf/www/ChenCJYDH23}, DADE~\cite{DBLP:journals/pvldb/DengCZWZZ24}, PEOs~\cite{DBLP:conf/icml/Lu0I24}, DDC~\cite{DBLP:conf/icde/YangLJZWSJW25}, and PDX~\cite{10.1145/3725333}, have been developed to optimize distance computation. We incorporate DADE to design an estimation-based filtering method, and we can explore integrating other distance computation optimization methods in future for further improvement.

\section{Conclusion}
In this paper, we proposed SkipDisk, a novel SSD-memory hybrid ANN search method. We assign dedicated pivots to each point to compute a tighter lower bound, reducing disk accesses. And we design an estimation-based filtering method based on this lower bound to enable effective filtering with reduced memory footprint. Second, we designed dimension-level pruning, bit-level pruning, and point-level pruning strategies to reduce memory usage while maintaining low latency. Third, we decoupled in-memory candidate generation from disk access by storing neighbor nodes in memory and used asynchronous I/O to hide disk access latency. Through these optimizations, we achieved lower latency performance than HNSW while significantly reducing memory footprint.


\bibliographystyle{ACM-Reference-Format}
\bibliography{sample}

\end{document}